\documentclass[
  journal=jog,
  manuscript=article-type,
  year=2024,
  volume=37,
]{cup-journal}

\usepackage{amsmath}
\usepackage[nopatch]{microtype}
\usepackage{booktabs}
\usepackage{enumitem}
\usepackage{rotating}
\usepackage{booktabs}
\usepackage[utf8]{inputenc}
\usepackage{enumerate}
\usepackage{subcaption}
\usepackage{pythonhighlight}
\usepackage{epsfig}
\usepackage{amsmath}
\usepackage{mathrsfs}  
\usepackage{amsthm}
\newtheorem{theorem}{Theorem}[section]

\newtheorem{lemma}[theorem]{Lemma}
\usepackage{graphicx}
\usepackage{xcolor}
\usepackage{float}
\usepackage{color}
\usepackage{bm}
\usepackage{graphicx}
\usepackage{multirow}

\renewcommand{\d}{\mathrm{d}}
\newcommand{\Ki}{\mathrm{k}}
\newcommand{\leftstate}{l}
\newcommand{\rightstate}{r}
\newcommand{\Kh}{K}

\newcommand{\HtwoO}{\mathrm{H}_2\mathrm{O}}

\usepackage{amsfonts}
\usepackage{amsmath}

\makeatletter
\newcommand*{\da@rightarrow}{\mathchar"0\hexnumber@\symAMSa 4B }
\newcommand*{\da@leftarrow}{\mathchar"0\hexnumber@\symAMSa 4C }
\newcommand*{\xdashrightarrow}[2][]{%
  \mathrel{%
    \mathpalette{\da@xarrow{#1}{#2}{}\da@rightarrow{\,}{}}{}%
  }%
}
\newcommand{\xdashleftarrow}[2][]{%
  \mathrel{%
    \mathpalette{\da@xarrow{#1}{#2}\da@leftarrow{}{}{\,}}{}%
  }%
}
\newcommand*{\da@xarrow}[7]{%
  \sbox0{$\ifx#7\scriptstyle\scriptscriptstyle\else\scriptstyle\fi#5#1#6\m@th$}%
  \sbox2{$\ifx#7\scriptstyle\scriptscriptstyle\else\scriptstyle\fi#5#2#6\m@th$}%
  \sbox4{$#7\dabar@\m@th$}%
  \dimen@=\wd0 %
  \ifdim\wd2 >\dimen@
    \dimen@=\wd2 %
  \fi
  \count@=2 %
  \def\da@bars{\dabar@\dabar@}%
  \@whiledim\count@\wd4<\dimen@\do{%
    \advance\count@\@ne
    \expandafter\def\expandafter\da@bars\expandafter{%
      \da@bars
      \dabar@ 
    }%
  }%
  \mathrel{#3}%
  \mathrel{%
    \mathop{\da@bars}\limits
    \ifx\\#1\\%
    \else
      _{\copy0}%
    \fi
    \ifx\\#2\\%
    \else
      ^{\copy2}%
    \fi
  }%
  \mathrel{#4}%
}
\makeatother

\renewcommand{\d}{\textrm{d}}

\usepackage{amsmath}
\usepackage[nopatch]{microtype}
\usepackage{booktabs}

\title{A unified kinematic wave theory for melt infiltration into firn}
\author{Mohammad Afzal Shadab$^{1,2}$}
\author{Anja Rutishauser$^3$} 
\author{Cyril Grima$^2$} 
\author{Marc Andre Hesse$^4$}

\affiliation{$^1$ Oden Institute for Computational Engineering and Sciences, University of Texas at Austin, USA\\
$^2$ University of Texas Institute for Geophysics, The University of Texas at Austin, USA \\$^3$ Department of Glaciology and Climate, Geological Survey of Denmark and Greenland, Denmark\\ $^4$ Department of Earth and Planetary Sciences, Jackson School of Geosciences, University of Texas at Austin, USA}
\email[Mohammad Afzal Shadab]{mashadab@utexas.edu}


\begin{document}

\begin{abstract}
Motivated by the refreezing of melt water in firn we revisit the one-dimensional percolation of liquid water and non-reactive gas in porous ice. We analyze the dynamics of infiltration in the absence of capillary forces and heat conduction to understand the coupling between advective heat and mass transport in firn. In this limit, we formulate a kinematic wave theory 
that results in a 2$\times$2-system of hyperbolic partial differential equations (PDEs) corresponding to the conservation of composition and enthalpy. For simple initial conditions (Riemann problems) this system admits self-similar solutions that illuminate the structure of melting/refreezing fronts and analytical solutions are provided for 12 basic cases of physical relevance encountered in the literature. Further we develop an extended kinematic theory that encompasses the cases when the firn saturates completely to form a perched water table governed by elliptic PDE so that the model is no longer fully hyperbolic (local). These solutions provide benchmarks for numerical models of melt infiltration into firn. They also provide insight into important physical processes such as the formation of frozen fringes, the perching of melt water on pre-existing low porosity layers and the conditions required for impermeable ice lens formation. Lastly, these analytic solutions can be utilized to improve and compare the performance of the firn hydrology, ice-sheet and Earth system models.
Our analysis provides a theoretical framework to understand these important processes in firn which affect the partitioning between meltwater infiltration and surface runoff and therefore determine the surface mass loss from ice sheets and its contribution to sea level rise.





\end{abstract}

\section{Introduction}

Modeling the infiltration of water from both rainfall and surface melting into firn is still a challenging problem \citep{firn2024firn}. A better understanding of this process is needed to improve its representation in the firn hydrology, ice sheet system and Earth system models \citep{clark2017analytical}. 
An understanding of firn processes can be developed through the analysis of simplified but physics-based models. Such models are fast and cost-efficient but they only provide approximate results. Furthermore, the analytical results for synthetic test problems are useful for the verification and validation of more complex firn hydrological models that currently show significant deviations \citep{clark2017analytical,vandecrux2020firn}. 
These analytic solutions can also help to reduce the uncertainties in the existing models and numerical algorithms such as bucket-schemes to match the field observations by introducing more physics-based constraints \citep{vandecrux2020firn,stevens2020community,ashmore2020meltwater}. 

\citet{colbeck1971one,colbeck1972theory} and \citet{colbeck1973water} spearheaded the mathematical modeling of melt infiltration into temperate snow by adapting models for two-phase flow in porous media to snow and firn \citep{Scheidegger1960,Morel-Seytoux1969}. \citet{colbeck1974capillary} demonstrates that flow due to gravity is the dominant process while capillary suction/diffusion plays a minor role. Neglecting capillary diffusion results in a non-linear kinematic wave model that describes the evolution of the water/melt saturation in the firn \citep{colbeck1972theory}. This kinematic wave model results in a hyperbolic partial differential equation that allows analytic solutions using the method of characteristics (MOC) \citep{lighthill1955kinematic}. These analytic solutions capture the main features of field data from infiltration in the Seward glacier firn on the St. Elias mountains in Canada \citep{sharp1951meltwater}. 

Later Colbeck extended this kinematic theory to analyze the effects of impermeable basal boundary \cite{colbeck1974water}, the retention of water in snow \cite{colbeck1976analysis}, studied the effects of layering and heterogeneity on melt percolation \cite{colbeck1979water,Colbeck1991} and determined the permeability of snow using lysimeter data \citep{colbeck1982permeability}. Subsequently,  \cite{singh1997kinematic} use the kinematic theory to explore the interaction of drying and wetting fronts and then propose the cases of constant, increasing, and decreasing (drainage) rainfall rates with time \cite[for review see also][]{singh1997kinematicbook}.

\cite{colbeck1976analysis} also added thermodynamics to the kinematic theory to study infiltration into cold firn that requires refreezing of the melt at the wetting front. He shows that the retardation of the wetting front is relatively minor in snow because the latent heat of fusion is large. Further, \cite{clark2017analytical} develop analytic solutions for the kinematic model that describe the interaction of the wetting and drying fronts in the cold and temperate snow, following similar solutions for soils \citep{Charbeneau1984}. They provide both vertical profiles of water saturation in the snow as well as the outflow hydrograph and present detailed comparisons with numerical solutions. 


Variably saturated flows due to infiltration in firn have not been considered in the kinematic wave models. The conventional kinematic wave theory for infiltration fails in a fully-saturated region as the model equation is not valid anymore. Therefore, the kinematic wave theory needs an extension to capture fully-saturated regions. Perched water tables may likely form when the low porosity, initial or refrozen (also known as frozen fringe \citep{miller1972freezing}) medium is unable to accommodate the volumetric flux of water. \cite{shadab2022analysis} proposes a kinematic wave theory for transitional infiltration in heterogeneous temperate soil (or firn) which leads to a rising perched water table. It is a crucial pre-requisite to understand the formation of frozen fringe or ice lens since the porosity may drop enough to cause complete saturation. 

{Recently, there has been a consistent effort to find ice lenses (<10 cm thick), layers (10 cm-1 m thick) and slabs (>1 m thick) in firn from radar based measurements \citep{culberg2021extreme}, firn core studies \citep{samimi2020meltwater} and satellite data \citep{miller2022empirical}.} These impermeable barriers greatly affect the surface mass balance by making the lower firn inaccessible to the meltwater and therefore, increase the runoff especially in the percolation zone \citep{tc-9-1203-2015,Machguth2016,culberg2021extreme}. However, there is still a dearth of literature on understanding the mechanism of their formation and evolution. 
A better understanding of the meltwater infiltration process will help constrain its ultimate partitioning into refreeze, liquid storage, and runoff \citep{meyer2017continuum}. Additionally, this improvement will facilitate the development of more accurate physics-based models and enable comparison with existing firn hydrology models that currently exhibit significant deviations \citep{vandecrux2020firn}. This comparison will be achieved through verification and validation against analytic solutions for simple benchmark problems.

{Therefore in this paper we propose a unified kinematic wave theory for the infiltration of liquid water and gas in porous ice in the limit of negligible capillary forces by coupling the mass and enthalpy balances. We extend the theory to account for the formation and evolution of perched water table where hyperbolic nature of PDEs fails. It can help construct simple analytical solutions which can prove to be useful for infiltration in temperate as well as cold firn where the medium can also form fully saturated regions possibly due to the formation/existence of a frozen fringe. Finally, we derive conditions for an impermeable ice lens formation due to meltwater advection and refreezing. The remainder of this paper is divided into four sections. Section \ref{sec:model-formulation} presents the model formulation. Section \ref{sec:melt-transport} considers the problem of melt transport across a discontinuity and documents twelve nature-inspired cases with their analytical solutions. Section \ref{sec:validation} applies this theory to study a multilayered firn leading to formation of a perched firn aquifer and validates it with the numerical solution. Finally, Section \ref{sec:conclusion} concludes the paper.}

\section{Continuum Model Formulation} \label{sec:model-formulation}
We consider a one-dimensional transport problem of three phase, two component system comprising of water (H$_2$O) in liquid ($w$) and solid ice ($i$) phases along with non-reactive gas phase ($g$). Ice is considered to be a porous material when its volume fraction $\phi_i\in (0,1)$, where volume fraction $\phi_\alpha$ is the ratio of volume (m$^3$) of the phase $\alpha \in \{w,i,g \}$ to the representative elemental volume (m$^3$). In this section, we first define the conserved quantities, then introduce the governing equations and constitutive models, and finally provide the resulting dimensionless continuum model. The related assumptions will be introduced in this work when required.

\subsection{Conserved Quantities}
The first conserved variable is the  water ($\HtwoO$) composition (kg/m$^3$), $C$, defined as the total mass of water components per unit representative elemental volume (REV) as
\begin{align} \label{eq:C-def-mole-fraction}
    C := \rho_i \phi_i X_{\HtwoO,i} + \rho_w \phi_w X_{\HtwoO,w} + \rho_g \phi_g X_{\HtwoO,g}
\end{align}
where $\rho_\alpha$ is the density (kg/m$^3$), $\phi_\alpha$ refers to volume fraction of the phase $\alpha \in \{w,i,g \}$. The symbol $X_{\HtwoO,\alpha}$ represents the mole fraction of $\HtwoO$ component present in phase $\alpha$. Assuming that water vapor component is negligible in the gas phase ($X_{\HtwoO,g}=0$) and moreover ice and water phases are considered pure, i.e., $X_{\HtwoO,i}= X_{\HtwoO,w}=1$, we can simplify Equation \eqref{eq:C-def-mole-fraction} as 

\begin{align} \label{eq:C-working-def}
    C = \rho_i \phi_i + \rho_w \phi_w.
\end{align}
Assuming same density for ice and water, $\rho_i\approx \rho_w = \rho$, Equation \eqref{eq:C-working-def} further simplifies to
\begin{align} \label{eq:C-simplified-def}
    C = \rho(\phi_i +\phi_w ).
\end{align}
Next, the volume constraint gives
\begin{align}
    \phi_i+\phi_w+\phi_g&=1 \quad \textrm{and} \label{eq:vol-constraint}\\
    \Rightarrow C &= \rho (1-\phi_g) \label{eq:C-gas}.
\end{align}

The second conserved variable is related to the thermodynamics of the system. Since phase change is involved when ice melts, temperature does not accurately represent all states of the system considered \citep{aschwanden2012enthalpy,alexiades2018mathematical}. Therefore the second conserved variable is chosen to be the enthalpy of the system (J/m$^3$), $H$, which is defined as
\begin{align} \label{eq:5new}
    H &:= \rho_i \phi_i h_i(T) + \rho_w \phi_w h_w(T) + \rho_g \phi_g h_g(T),
\end{align}
where $h_\alpha$ is the specific enthalpy (J/kg) of the phase $\alpha$, which is a piecewise linear function of temperature (K), $T$. For simplicity we fix the reference enthalpy at the solidus to be $H=0$ where the system is at the melting temperature, $T=T_m$. The specific enthalpy $h_\alpha$ of each phase $\alpha \in \{w,i,g\}$ can then be defined as

\begin{align}
    h_i(T) &= \begin{cases}c_{p,i}(T-T_m),& T < T_m \quad (\textrm{or } H<0) \\0, & T \geq T_m \quad(\textrm{or } H \geq 0) \end{cases}, \label{eq:sp-enthalpy-ice} \\
    h_w(T) &= \begin{cases}0,& T < T_m \quad(\textrm{or } H<0) \\c_{p,w}(T-T_m)+L, & T \geq T_m \quad(\textrm{or } H\geq0)\end{cases} \label{eq:specific-enthalpy-formulation} \quad \text{and} \\ 
    h_g(T) &= c_{p,g}(T-T_m).
\end{align}
Here $c_{p,\alpha}$ is the specific heat capacity at constant pressure (J/kg$\cdot$K) for phase $\alpha$, $T_m$ is the melting temperature (K) and $L$ is the latent heat of fusion of water (J/kg). The density and specific heat capacity of gas are much lower than those of liquid water or ice, i.e., $\rho_g \ll \rho$ and $c_{p,g}<c_{p,i}$ or $c_{p,w}$ (see Table \ref{table:1}). It helps make a simplification that the gas phase contribution to the total enthalpy of the system is negligible. After substituting Equations \eqref{eq:sp-enthalpy-ice} and \eqref{eq:specific-enthalpy-formulation} in Equation \eqref{eq:5new}, $H$ can be ultimately formulated as

\begin{align} \label{eq:enthalpy-def}
    H &= \begin{cases}\rho c_{p,i} \phi_i (T-T_m), & {T < T_m} \quad (\textrm{or } H \leq 0) \\  \rho \phi_w L,  &{T= T_m} \quad (\textrm{or } 0 < H < CL) \\ \rho \phi_w \left(c_{p,w} (T-T_m) + L \right), &{T> T_m} \quad (\textrm{or } H \geq C L) \end{cases}.
\end{align}
Here the maximum enthalpy limit for the three-phase region is $CL$ because it is the enthalpy of the system at liquidus (see Figure \ref{fig:combined-variables}\emph{a}). The boundaries of the three-phase region are not included in the region $0 < H<CL$ because it strictly refers to the three phase region. However, no ice (no liquid water) state can exist at liquidus (solidus), which occurs at $T=T_m$. From the above formulation we classify three regions, where region 1 ($H \leq 0$) comprises of ice and gas, region 2 ($0 < H<CL$) contains all three phases and region 3 ($H \geq CL$) corresponds to no-matrix state consisting of only water and gas phases. This classification of regions will be referred to in the subsequent sections of this paper. In reality, the solid skeleton breaks down when porosity is less than a critical threshold and the two-phase medium becomes a slurry of non-contiguous solid grains suspended in liquid \citep{katz2022dynamics}.

The temperature and volume fractions of water, ice and gas phases can be evaluated from composition, $C$, and enthalpy, $H$, as shown in Figures~\ref{fig:combined-variables}\emph{a}-\ref{fig:combined-variables}\emph{d} respectively. As shown in Equation \eqref{eq:C-gas}, the volume fraction of gas, $\phi_g$, only depends on the composition, $C$, as shown in Figure \ref{fig:combined-variables}\emph{d}. Next we formulate the governing equations for this model corresponding to the two conserved variables.

\begin{figure}
    \centering
    \includegraphics[width=\linewidth]{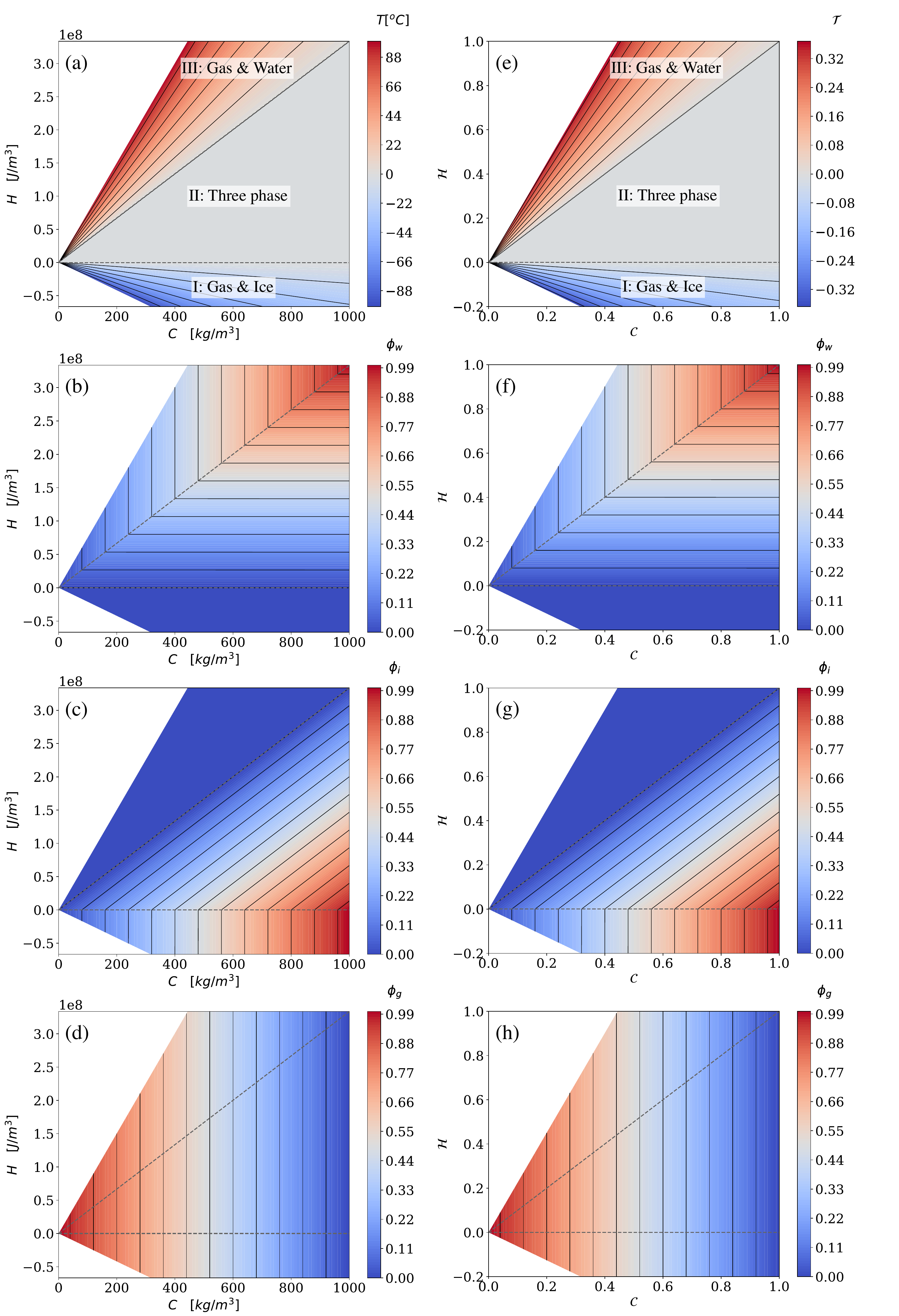}
    \vspace{-0.4cm}
    \caption{The dependence of temperature and volume fractions on dimensional and dimensionless enthalpy and composition, ($C,H$) and ($\mathcal{C},\mathcal{H}$) respectively. Dimensional $C,H$: (a) temperature and volume fractions of (b) water, (c) ice and (d) gas phases. Dimensionless $\mathcal{C},\mathcal{H}$: (e) scaled temperature and volume fractions of (f) water, (g) ice and (h) gas phases. The contours are restricted to $T\in[-100^\circ\textrm{C},100^\circ\textrm{C}]$ to avoid phase change at boiling as well as keep the contour levels consistent. Solid black lines are the level-sets whereas the dashed lines show the boundaries of the regions, i.e., solidus ($H=0$ or $\mathcal{H}=0$) and liquidus ($H=CL$ or $\mathcal{H}=\mathcal{C}$).}
    \label{fig:combined-variables}
\end{figure}

\subsection{Transport Model}
In a reference frame moving with ice, the conservation equations for water composition and system enthalpy are respectively given as

\begin{align}
    \frac{\partial C}{\partial t} + \nabla \cdot (\textbf{q} \rho X_{\HtwoO,w}+\textbf{q}_g \rho_g X_{\HtwoO,g}-\phi \rho D \nabla X_{\HtwoO,w}) &= 0 \quad \textrm{and} \label{eq:comp-conservation}\\
    \frac{\partial H}{\partial t} + \nabla \cdot (\textbf{q} \rho h_w - \overline{\kappa} \nabla T) &= 0, \label{eq:enthalpy-conservation}
\end{align}
where $\textbf{q}$ ($\textbf{q}_g$) is the volumetric flux of water (gas) phase (m$^3$/m$^2\cdot$s) relative to ice phase. Diffusion in the solid has been neglected in Equation \eqref{eq:comp-conservation}. Additionally, $D$ is the effective hydrodynamic dispersion of the porous medium (m$^2$/s). The effective thermal conductivity of the mixture (W/m$\cdot$K), $\overline{\kappa}$, is defined as a volume fraction weighted sum of thermal conductivities (W/m$\cdot$K), $\kappa_\alpha$, of the individual phases $\alpha$ as

    \begin{equation}
        \overline{\kappa} = \kappa_w \phi_w + \kappa_i \phi_i + \kappa_g \phi_g, 
    \end{equation}
where gas component can again be neglected since gas is relatively a thermal insulator as $\kappa_g \ll \kappa_w$ or $\kappa_i$ (see Table \ref{table:1}). Therefore, the effective thermal conductivity can be simplified as

\begin{equation}
        \overline{\kappa} = \kappa_w \phi_w + \kappa_i \phi_i. \label{eq:13}
    \end{equation}
Since $X_{\HtwoO,g}=0$ and $X_{\HtwoO,w}=1$, Equation \eqref{eq:comp-conservation} can be simplified as  

\begin{align}
    \frac{\partial C}{\partial t} + \nabla \cdot (\textbf{q} \rho) &= 0 \label{eq:comp-conservation-final}
\end{align}
\subsection{Constitutive Relations}
The volumetric flux of water relative to ice, $\textbf{q}$, can be written using extended Darcy's law,

\begin{align} \label{eq:darcy-full}
    \textbf{q} = - \frac{\Ki(\varphi) k_{rw}(s)}{\mu} (\nabla p - \rho \textbf{g})
\end{align}
where $\Ki$ is the absolute permeability (m$^2$) which is a function of porosity $\varphi$ $(\varphi=\phi_w+\phi_g =1-\phi_i)$, the ratio of void volume to the bulk volume, $p$ is water pressure (Pa), $\mu$ is the viscosity of water (Pa$\cdot$s) and $\textbf{g}$ is the acceleration due to gravity vector (m$^2$/s). The relative permeability for multi-phase flow, $k_{rw}$, display complex hysteresis \citep{blunt2017multiphase}, but here we only consider the simplest case with power law dependence. The relative permeability of water phase, $k_{rw}$, is a function of the water saturation, $s$, which is the ratio of water phase volume to void volume $\left(s= \frac{\phi_w}{1-\phi_i}\right)$.  We assume that the water phase becomes immobile below a certain residual water saturation, $s_{w r}$. Similarly the gas phase becomes immobile below the residual gas saturation, $s_{gr}$. As a result, the two-phase fluid flow of both gas and water phases is restricted to regions where $s_{wr}<s<1-s_{gr}$. We will refer to regions with $s=1-s_{gr}$ as saturated in the remainder of this paper. The residual water saturation during drainage comes out to be $s_{w r}\sim0.07$ from lysimeter \citep{colbeck1976analysis} and calorimeter \citep{coleou1998irreducible} techniques. However, as the ice is water-wet with a near-zero contact angle at the ice-water-air interface \citep{knight1971experiments}, the residual water saturation during saturation rise (imbibition) is zero due to hysteresis in the relative permeability-capillary pressure curve \citep{carlson1981simulation,blunt2017multiphase}. The phenomenon of hysteresis has been neglected in the firn hydrology literature but it will affect the speeds of the meltwater fronts.

Next we assume the problem is gravity dominated in the unsaturated mediums such that the spatial variations in the difference of fluid phases' pressures (capillary pressure) is negligible at the problem length scales \citep{colbeck1972theory}. See \cite{Smith1983,shadab2022analysis,shadab2022hyperbolic} for a more detailed discussion on neglecting the pressure term in context of soils using scaling analysis. As a result, the pressure of the water phase in the unsaturated regions becomes a constant, equal to the reference gas pressure, i.e., $p=0$ \citep{colbeck1972theory,shadab2022analysis}. Plugging it in Equation \eqref{eq:darcy-full} eliminates the diffusive, pressure term. The volumetric flux of water, $\textbf{q}$, finally takes the gravity-driven form

\begin{align} \label{eq:darcy-law-simplified}
      \textbf{q} =  \frac{\Ki(\varphi) k_{rw}(s)}{\mu}  \rho \textbf{g}.
\end{align}
The absolute permeability of ice (m$^2$), $\Ki$, and the relative permeability of water, $k_{rw}$, are assumed to be power laws \citep{kozeny1927uber,carman1937fluid,brooks1964hydraulic,bear2013dynamics,meyer2017continuum} defined as

\begin{align}
\Ki(\varphi) &=\Ki_0 \varphi^m = \Ki_0 (1-\phi_i)^m , \label{eq:abs-perm-def} \\
k_{rw}(s) &= k^0_{rw} s^n = k^0_{rw} \left( \frac{\phi_w}{1-\phi_i}\right)^n , \label{eq:rel-perm-def}
\end{align}
where $\Ki_0$ is a model constant, which can be considered as an absolute permeability (m$^2$) when there no ice matrix, and $k^0_{rw}$ is end point relative permeability of water phase. Here we have also assumed that the residual saturations of both water and gas phases are zero, i.e., $s_{wr}=s_{gr}=0$. It will provide accurate speeds for the wetting fronts moving into dry firn, due to hysteresis in the relative permeability. Plugging Equations \eqref{eq:abs-perm-def} and \eqref{eq:rel-perm-def} in Equation \eqref{eq:darcy-law-simplified} finally gives

\begin{align}\label{eq:darcy-law-final}
      \textbf{q}(\phi_i,\phi_w) = \begin{cases}
      \textbf{0}, &H\leq 0 \\ \frac{\Ki_0 k^0_{rw}}{\mu}  \rho g (1-\phi_i)^m \left( \frac{\phi_w}{1-\phi_i}\right)^n \hat{\textbf{g}} = K_h (1-\phi_i)^m \left( \frac{\phi_w}{1-\phi_i}\right)^n \hat{\textbf{g}}, & 0<H<CL \textrm{ \& } {C<\rho}
      \end{cases}
\end{align}
where the acceleration due to gravity vector is $\textbf{g}=g\hat{\textbf{g}}$ with $\hat{\textbf{g}}$ being the unit vector in the direction of gravity. Moreover, $\textbf{0}$ is a zero vector in three-dimensions, i.e., $\textbf{0}=\left[ 0,0,0\right]^T$. The symbol $K_h = \frac{\Ki_0 k^0_{rw}}{\mu}\rho g$ is a model constant equivalent to the saturated hydraulic conductivity of water at unity porosity (m/s). It is a known constant which can be considered as the maximum gravity-dependent volumetric flux of water, $|\textbf{q}|$, at unity porosity. Note that the dynamics at unity porosity is not Darcy-type as the constitutive relationships (\ref{eq:darcy-full}-\ref{eq:abs-perm-def}) are only valid for a porous medium. Therefore we will restrict our analysis to the porous media where $\varphi < 1$.

Next we will non-dimensionalize the model in order to make it scale independent and also find dominant terms which govern the crucial physics of the problem.

\subsection{Scaling}
The model is scaled using dimensionless variables for composition, $\mathcal{C}$, enthalpy, $\mathcal{H}$, temperature, $\mathcal{T}$, depth, $\zeta$, and time, $\tau$, which are defined as

\begin{equation} \label{eq:20}
    \mathcal{C} = \frac{C}{\rho}, \mathcal{H} = \frac{H}{\rho L}, \hspace{1mm}\mathcal{T}=\frac{T-T_m}{T_m},\hspace{1mm} \zeta = \frac{z}{\delta}, \hspace{1mm}\text{ and }\hspace{1mm} \tau = \frac{t K_h}{\delta}.
\end{equation}
Here the spatial coordinates (for example, the depth coordinate $z$) are non-dimensionalized by length scale of heterogeneity or the REV scale of the problem, $\delta$. Time variable is scaled by the shortest time of water seepage across the characteristic length through a medium with unity porosity, i.e., $\delta /K_h$. The definitions of conserved quantities $C$ and $H$, given in Equations (\ref{eq:C-simplified-def}) and (\ref{eq:enthalpy-def}) respectively, thus transform into the dimensionless forms

\begin{align}
    \mathcal{C} &= \phi_i+\phi_w = 1-\phi_g , \label{eq:dimensionless-composition-formulation}\\
    \mathcal{H} &= \begin{cases} \mathcal{C} \hspace{0.5mm} \mathrm{Ste} \hspace{0.5mm} \mathfrak{c}_{p,r}  \mathcal{T} , & \mathcal{H} \leq 0 \\
      \phi_w, &0< \mathcal{H}< \mathcal{C} \\
    \mathcal{C} (\mathrm{Ste} \hspace{0.3mm} \mathcal{T} + 1), & \mathcal{H} \geq \mathcal{C} \end{cases}  \label{eq:dimensionless-enthalpy-formulation},
\end{align}
where Ste is the Stefan number defined as ratio of sensible heat of water at melting temperature to the latent heat of fusion of $\HtwoO$, i.e., Ste$=c_{p,w} T_m / L$ and $\mathfrak{c}_{p,r}=c_{p,i}/c_{p,w}$ is the ratio of specific heat of ice to that of water. From the formulations of dimensionless enthalpy \eqref{eq:dimensionless-enthalpy-formulation} and dimensional specific enthalpy of water phase \eqref{eq:specific-enthalpy-formulation}, the dimensionless temperature, $\mathcal{T}$, and dimensionless specific enthalpy of water phase, $\mathfrak{h}_w=h_w/L$, can be derived as

\begin{align}
    \mathcal{T}(\mathcal{C},\mathcal{H}) &= \begin{cases} \frac{\mathcal{H}}{ \mathcal{C} \hspace{0.3mm} \textrm{Ste} \hspace{0.5mm \mathfrak{c}_{p,r}}}, &{\mathcal{H} \leq 0} \\
    0, &0<\mathcal{H}< \mathcal{C}\\
    \frac{1}{\textrm{Ste}}\left(\frac{\mathcal{H}}{\mathcal{C}}-1 \right), &{\mathcal{H} \geq \mathcal{C}}
    \end{cases}, \qquad \textrm{ and } \quad
     \mathfrak{h}_w (\mathcal{C},\mathcal{H}) = \begin{cases} 0,& \mathcal{H} \leq 0 \\ 1 , & 0 < \mathcal{H}<\mathcal{C}\\ \frac{\mathcal{H}}{\mathcal{C}}, & \mathcal{H} \geq \mathcal{C} \\\end{cases}. \label{eq:1.78hw}
\end{align}
Subsequently the volume fractions of the phases, $\phi_\alpha$, and the porosity of the medium, $\varphi$, can be rewritten as functions of $\mathcal{C}$ and $\mathcal{H}$ as

\begin{align}
     \phi_w(\mathcal{C},\mathcal{H}) &= \begin{cases}0,& \mathcal{H} \leq 0 \\ \mathcal{H} , & 0 < \mathcal{H}<\mathcal{C}\\ \mathcal{C}, & \mathcal{H} \geq \mathcal{C} \\\end{cases}, \quad
     \phi_i(\mathcal{C},\mathcal{H}) = \begin{cases}\mathcal{C},& \mathcal{H} \leq 0 \\ \mathcal{C}-\mathcal{H} , & 0 < \mathcal{H}<\mathcal{C}\\ 0, & \mathcal{H} \geq \mathcal{C} \\\end{cases}, \quad
     \phi_g(\mathcal{C}) = 1 -  \mathcal{C} \label{eq:1.78phig} \\
     \textrm{ and } \quad \varphi(\mathcal{C},\mathcal{H}) &= \begin{cases}1-\mathcal{C},& \mathcal{H} \leq 0 \\ 1-\mathcal{C}+\mathcal{H} , & 0 < \mathcal{H}<\mathcal{C}\\ 1, & \mathcal{H} \geq \mathcal{C} \\\end{cases}.   \label{eq:27}
\end{align}

The scaled temperature and volume fractions of water, ice and gas phases can be evaluated from dimensionless composition, $\mathcal{C}$, and dimensionless enthalpy, $\mathcal{H}$, as shown in Figures \ref{fig:combined-variables}\emph{e}-\ref{fig:combined-variables}\emph{h} respectively. As shown in Equation \eqref{eq:C-gas}, the volume fraction of gas, $\phi_g$, only depends on the dimensionless composition as illustrated in Figure \ref{fig:combined-variables}\emph{h}.

The composition and enthalpy transport equations (\ref{eq:comp-conservation-final} and \ref{eq:enthalpy-conservation}) thus take the dimensionless form
\begin{align}
      \frac{\partial \mathcal{C}}{\partial \tau} + \nabla \cdot \left((1-\phi_i)^m \left( \frac{\phi_w}{1-\phi_i}\right)^n \hat{\textbf{g}} \right) = 0, \label{eq:dimless-comp-conservation-final} \\
     \frac{\partial \mathcal{H} }{\partial \tau} + \nabla \cdot \left(\mathfrak{h}_w (1-\phi_i)^m \left( \frac{\phi_w}{1-\phi_i}\right)^n \hat{\textbf{g}} -  ( \phi_w + \kappa \phi_i) \frac{\textrm{Ste}}{Pe_{\mathcal{H}}}\nabla \mathcal{T} \right) = 0. \label{eq:dimless-enthalpy-conservation-final}
\end{align}
Here $\kappa= \kappa_i/\kappa_w$ is the ratio of the heat conductivities of ice to water and $\alpha_T={\kappa_w}/{\rho c_{p,w}}$ (m$^2$/s) is the thermal diffusivity of water. The ratio of heat convected to heat diffused is defined as the Peclet number for enthalpy equation, $Pe_{\mathcal{H}}=K_h \delta/\alpha_T$. Moreover, the divergence and gradient operators are now scaled with inverse of characteristic depth, $1/\delta$.

Using thermodynamic and fluid flow parameters from Table \ref{table:1}, the value of Peclet number, $Pe_{\mathcal{H}}$, comes out to be 3758.6 for $\delta=1$ m. This value indicates a three orders of magnitude higher heat convection compared to heat diffusion for gravity-driven infiltration in firn. Therefore, we can neglect the second order heat conduction term in Equation \eqref{eq:dimless-enthalpy-conservation-final}, which is also the necessary condition in the three phase region ($0 < \mathcal{H}<\mathcal{C}$ or $\mathcal{T}=0$) as $\nabla \mathcal{T} = 0$. Assuming local thermodynamic equilibrium (LTE), in the limit $Pe_{\mathcal{H}}\to \infty$, the system of dimensionless governing equations (\ref{eq:dimless-comp-conservation-final} and \ref{eq:dimless-enthalpy-conservation-final}) then reduces to quasi-linear system of coupled hyperbolic equations,
\begin{align} \label{eq:dimless-system-gov-eqs}
    \frac{\partial \textbf{u}}{\partial \tau} + \nabla \cdot \textit{\textbf{f}}(\textbf{u})= \textbf{0}
\end{align}
where $\textbf{u} = [\mathcal{C},\mathcal{H} ]^T$ is the vector of dimensionless conserved variables and $\textit{\textbf{f}}(\textbf{u}) = [ \textit{\textbf{f}}_\mathcal{C},\textit{\textbf{f}}_\mathcal{H}]^T$ is the vector of their corresponding nonlinear flux vectors. Here the flux vector functions for the dimensionless composition and enthalpy are given as 

\begin{align} \label{eq:dimless-flux-CH-vector}
    \textit{\textbf{f}}:=\textit{\textbf{f}}_{\mathcal{C}} =  \textit{\textbf{f}}_{\mathcal{H}} = \begin{cases} \textbf{0},& \mathcal{H} \leq 0, \\ (1-\mathcal{C}+\mathcal{H})^m \left( \frac{ \mathcal{H}}{1-\mathcal{C}+\mathcal{H}}\right)^n \hat{\textbf{g}}, & 0 < \mathcal{H}<\mathcal{C}\end{cases}.
\end{align}

\begin{figure}
\centering
  \includegraphics[width=0.5\linewidth,trim = 0 0 0 0cm,clip]{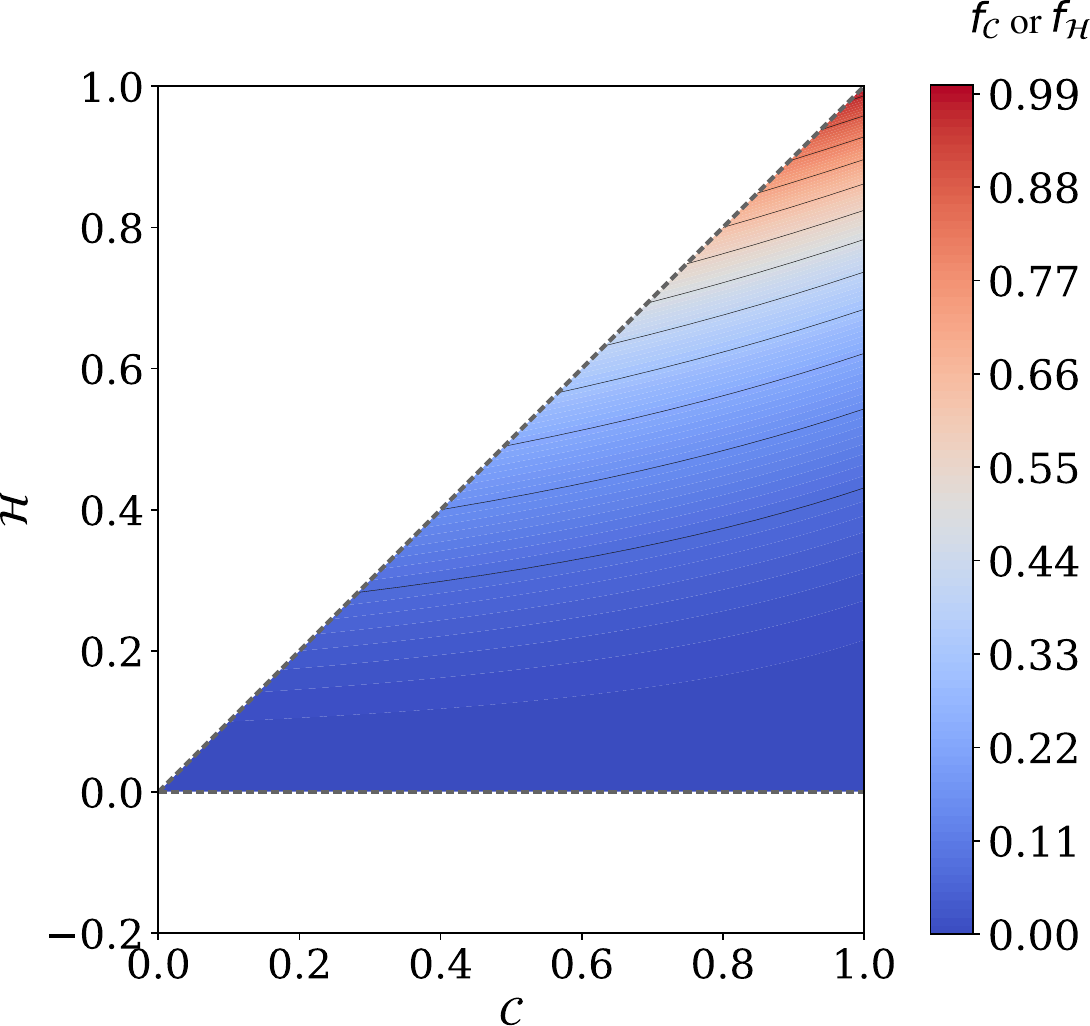}
\caption{The dimensionless flux of composition or enthalpy in $\mathcal{C}\mathcal{H}$ phase space for $m=3$ and $n=2$. In region 1 consisting of water and gas region ($\mathcal{H} \leq 0$) as well as region 3 comprising of three phase region ($0 < \mathcal{H}<\mathcal{C}$), the flux of dimensionless enthalpy and composition are identical, i.e., $f_\mathcal{C}=f_\mathcal{H}$. Region 3 with water and gas ($\mathcal{H}\geq \mathcal{C}$) is not considered in the present work.}
\label{fig:fluxes}
\end{figure}

The above analysis shows the distinct system behaviors in the different regions based on fluxes (see Figure \ref{fig:fluxes}). In this work, region 3 with no solid matrix is not considered since the constitutive relation for volumetric flux (Darcy's law) is not valid anymore. From Equation \eqref{eq:dimless-flux-CH-vector}, it can be observed that the fluxes of enthalpy and composition in the system of governing equations \eqref{eq:dimless-system-gov-eqs} are identical in regions 1 and 2 defined by the symbol $\textit{\textbf{f}}$ for brevity. This is the result of scaling owing to the fact that the composition changes only when water infiltrates or convects while carrying the enthalpy in form of latent heat (and specific heat) along with it.

In the next section, we will consider a simple problem of melt transport across a discontinuity to utilize the method of characteristics (MOC) for solving the system of hyperbolic partial differential equations \citep{lighthill1955kinematic,leveque1992numerical} given in Equation \eqref{eq:dimless-system-gov-eqs}.

\begin{table}
\caption{A summary of simplified thermodynamic properties as well as flow properties of water in porous ice used in present work}
\centering
\begin{tabular}{lcc}
\toprule
 Parameter & Value & Units   \\
 \midrule
 $\rho$& 1000 & kg/m$^3$ \\
 $c_{p,w}$& 4186 &  J/(kg K)  \\
  $c_{p,i}$& 2106.1 & J/(kg K)    \\
$\kappa_w$&0.606& W/(m K)\\
$\kappa_i$&2.25& W/(m K)\\
$L$&333.55& kJ/kg\\
$T_m$&273.16& K\\
$\alpha_T$&1.45$\times 10^{-7}$& m$^2$/s\\
$\mathfrak{c}_{p,r}$&0.503&-\\
Ste&3.428&-\\
 \hline
  $\Ki_0$& $5.56\times10^{-11}$ \citep{meyer2017continuum} & m$^2$ \\
 $k_{rw}^0$& 1.0 &  -  \\
  $m$& 3.0 &  -  \\
   $n$& 2.0 &  -  \\
  $g$& 9.81 &m/s$^2$    \\
$\mu$&$10^{-3}$&Pa s\\
$K_h$& $5.45 \times 10^{-4}$ &m/s\\
\bottomrule
\end{tabular}\label{table:1}
\end{table}

\section{Melt transport across a discontinuity}\label{sec:melt-transport}

\begin{figure}
    \centering
    \includegraphics[width=0.8\linewidth]{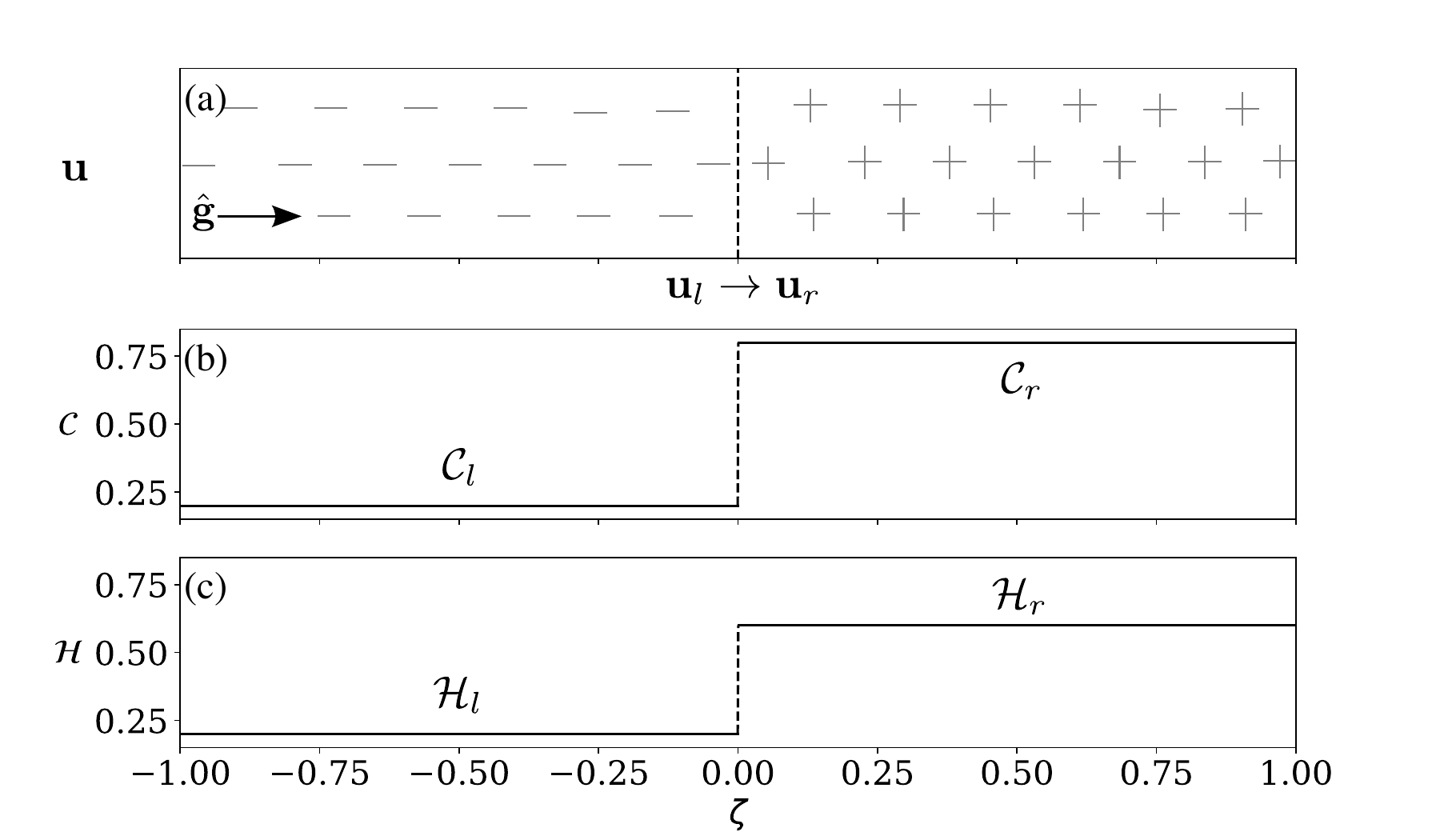}
    \caption{One-dimensional Riemann problem: (a) Schematic representation of $\textbf{u}$ across a discontinuity within or at the boundary of a porous firn. Initial conditions for the Riemann problem for conserved variables (b) $\mathcal{C}$ and (c) $\mathcal{H}$ plotted against dimensionless depth coordinate, $\zeta$.}
    \label{fig:representative-figure}
\end{figure}

This section considers the reaction front arising from melt flow across a discontinuity in dimensionless composition and enthalpy at a depth, say $\zeta = 0$, as shown in Figure \ref{fig:representative-figure}. The crucial dynamics of such problems can be understood using hyperbolic analysis of the coupled system of partial differential equations in one dimension \citep{lighthill1955kinematic,leveque1992numerical,venkatraman2014analytical,jordan2015reactive,ghaderi2019theory}.

\subsection{General structure of reaction fronts}
Consider the following one-dimensional initial value problem with two constant states, known as Riemann problem. See \cite{leveque1992numerical} for a pedagogical introduction to Riemann problems and their analysis. Let the spatial dimension be the direction of gravity, $\hat{\textbf{g}}$, which aligns with the [dimensionless] depth coordinate, $z$ [$\zeta$]. In that case, the dimensionless flux vector \eqref{eq:dimless-flux-CH-vector} reduces to

\begin{align} \label{eq:dimless-flux-CH}
   f: = \begin{cases} 0,& \mathcal{H} \leq 0 \\ (1-\mathcal{C}+\mathcal{H})^m \left( \frac{ \mathcal{H}}{1-\mathcal{C}+\mathcal{H}}\right)^n, & 0 < \mathcal{H}<\mathcal{C}\end{cases}
\end{align}

The system of dimensionless composition and enthalpy conservation equations \eqref{eq:dimless-system-gov-eqs} can be written in one-dimensional depth coordinates, $\zeta$, as

\begin{align} \label{eq:dimless-system-gov-eqs-1D}
    \textbf{u}_{\tau} + \textbf{f}(\textbf{u})_{\zeta} = \textbf{0}, \quad \tau \in \mathbb{R}^{+}, \quad \zeta \in \mathbb{R},
\end{align}
with initial conditions

\begin{align}
    \textbf{u}=\begin{cases} \textbf{u}_\leftstate, &\zeta <0 \\  \textbf{u}_\rightstate, &\zeta >0 \end{cases}
\end{align}
where the flux vector in $\zeta$ direction is $\textbf{f}(\textbf{u})=[f,f]^T$ and the subscripts $\tau$ and $\zeta$ refer to the partial derivatives with respect to the dimensionless time and depth respectively. The symbol $\mathbb{R}$ refers to the field of real numbers $(-\infty,\infty)$ whereas $\mathbb{R}^+$ refers to the set of positive real numbers $(0,\infty)$. The subscripts $\leftstate$ and $\rightstate$ refer to the state of the system on the left and right sides of a discontinuity. An example of an initial discontinuity with left state, $\textbf{u}_l$, and right state, $\textbf{u}_r$, is shown in Figures \ref{fig:representative-figure}\emph{b} and \ref{fig:representative-figure}\emph{c}. The flow is towards the direction of gravity, assumed to be in $+  \zeta$-direction. The solution to the Riemann problem for well-behaved systems of two coupled nonlinear partial differential equations is characterized by the formation of an intermediate state, $\textbf{u}_i$, bounded by two waves $\mathscr{W}_1$ and $\mathscr{W}_2$ \citep{leveque1992numerical}. This solution structure, observed in Figure \ref{fig:Riemann-solution-schematic}, can be represented as

\begin{align}
 \textbf{u}_\leftstate\xrightarrow{\mathscr{W}_1}\textbf{u}_i \xrightarrow{\mathscr{W}_2}\textbf{u}_\rightstate 
\end{align}
\begin{figure}
    \centering
    \includegraphics[width=0.7\linewidth]{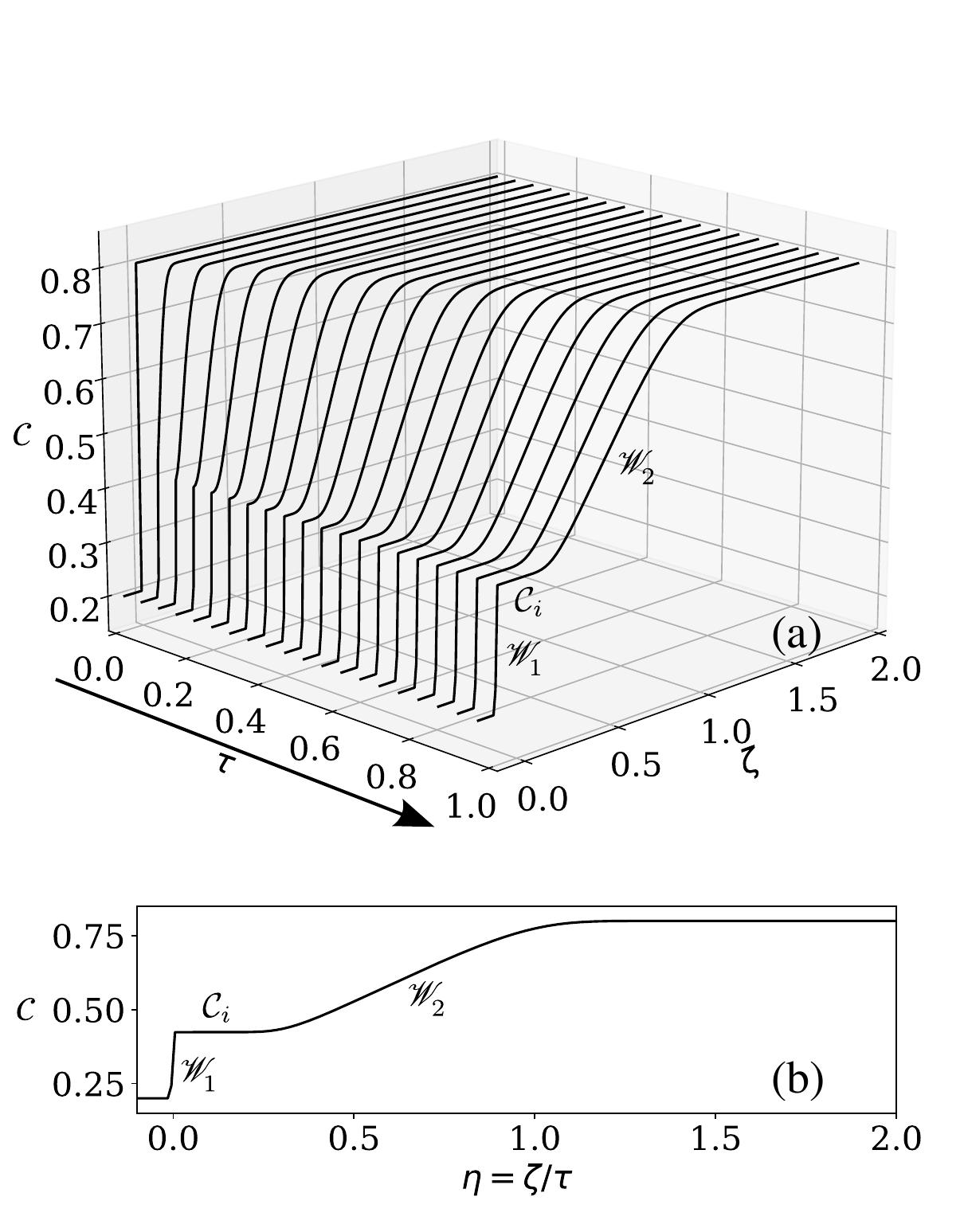}
    \caption{Solution of the Riemann problem introduced in Figure \ref{fig:representative-figure}. (a) Evolution of dimensionless composition, $\mathcal{C}$, in space, $\zeta$, and time, $\tau$. (b) The same self-similar solution plotted as a function of similarity variable $\eta=\zeta/\tau$.}
    \label{fig:Riemann-solution-schematic}
\end{figure}
In the context of reactive meltwater transport, the waves $\mathscr{W}_1$ and $\mathscr{W}_2$ are the reaction fronts and the intermediate state, $\textbf{u}_i$, corresponds to a state between the fronts. The system \eqref{eq:dimless-system-gov-eqs-1D} can be recasted into a quasilinear form using chain rule as

\begin{align} \label{eq:dimless-system-linearized}
    {\textbf{u}}_\tau + \nabla_{\textbf{u}} \textbf{f}(\textbf{u}) \hspace{0.5mm} \textbf{u}_\zeta = \textbf{0}
\end{align}
where $\nabla_{\textbf{u}} \textbf{f}(\textbf{u})$ is the gradient of flux, $\textbf{f}(\textbf{u})$, with respect to the conserved variables, $\textbf{u}$, which takes the matrix form

\begin{equation} \label{eq:dimless-Jacobian}
    \nabla_{\textbf{u}} \textbf{f}(\textbf{u}) = \begin{bmatrix} f_{,\mathcal{C}} & f_{,\mathcal{H}} \\ f_{,\mathcal{C}} & f_{,\mathcal{H}} \end{bmatrix}
\end{equation}
where the subscripts $,\mathcal{C}$ and $,\mathcal{H}$ refer to the partial derivatives with respect to dimensionless composition and enthalpy respectively. The derivatives of the flux gradient above can be evaluated explicitly, which are given in Appendix A. The system of advection equations \eqref{eq:dimless-system-linearized} results in waves (fronts) propagating with their characteristic velocities. These fronts have self-similar stretching patterns, because of their own characteristic velocities given by the flux gradient. 

\subsection{Self-similarity of reaction fronts}

The recognition of the constant stretching morphology of the reaction fronts from an initial jump condition allows the introduction of the similarity variable

\begin{equation}
    \eta = \frac{\zeta}{\tau}
\end{equation}
Physically, $\eta$ describes the dimensionless propagation velocity of the reaction front. The solution generally collapses into a single profile when plotted as a function of $\eta$ (see Figure \ref{fig:Riemann-solution-schematic}\textit{b}). Therefore, the system of partial differential equations \eqref{eq:dimless-system-linearized} can be transformed into a system of ordinary differential equations by considering the nonlinear eigenvalue problem 

\begin{align} \label{eq:37}
     (\textbf{A} - \lambda_p \textbf{I})~\textbf{r}_p = \textbf{0}, \quad p\in \{ 1,2 \} 
\end{align}
where the flux gradient is $\textbf{A}=\nabla_{\textbf{u}} \textbf{f}(\textbf{u})$ and the eigenvector is $\textbf{r}_p=\d \textbf{u}/\d \eta$ corresponding to the eigenvalue $\lambda_p$. Here the eigenvalues $\lambda_1$ and $\lambda_2$ are the characteristic propagation speeds of the waves $\mathscr{W}_1$ and $\mathscr{W}_2$ respectively. The associated eigenvectors, $\textbf{r}_p=\d \textbf{u}/\d \eta$, give the pathways through $\mathcal{C}\mathcal{H}$ plane, also referred to as hodograph plane, that satisfy the conservation equations (see Figure \ref{fig:3}\textit{b} for example).
Any constant state \textbf{u} that satisfies $\d \textbf{u}/\d \eta=0$ satisfies the conservation law \eqref{eq:dimless-system-linearized} trivially. Continuous solutions to Equation \eqref{eq:37} satisfy the conservation equations \eqref{eq:dimless-system-linearized}.

The opposite case is of the self-sharpening wave front instead of self-stretching as the information that propagates through the characteristics crosses each other \citep{leveque1992numerical}. Here the continuous solution fails and a moving or a stationary discontinuity appears where the local state changes abruptly. The propagation velocity of this discontinuity in the solution satisfies the Rankine-Hugoniot (R-H) jump condition \citep{leveque1992numerical}, which is derived from the discrete conservation, in this case of both mass and enthalpy, around the discontinuity as

\begin{align} \label{eq:RH-condition}
    \Lambda_{\mathscr{S}} (\textbf{u}_+ ,\textbf{u}_-) = \frac{\textbf{f}(\textbf{u}_+)-\textbf{f}(\textbf{u}_-)}{\textbf{u}_+-\textbf{u}_-} = \frac{[\textbf{f}(\textbf{u})]}{[\textbf{u}]}
\end{align}
where $\Lambda_\mathscr{S}$ is the shock speed and $[\,.\,]$ refers to the jump condition across the shock and subscripts $+$ and $-$ refer to the state on the left and right sides of a shock wave. Note that the left ($l$) and right ($r$) states might not necessarily be the left ($-$) and right ($+$) sides of a shock front, due to the presence of intermediate state(s). 

\subsection{Construction of solution in $\mathcal{C}\mathcal{H}$ hodograph plane}
The self-similar solutions are constructed by identifying directions in $\mathcal{C}\mathcal{H}$ hodograph plane that satisfy conservation laws and the equation of state. One such direction allows a continuous variation in \textbf{u}, which can be found by integrating the eigenvectors of the flux gradient. Another set of direction is determined by the nonlinear algebraic system of equations arising from R-H jump condition \eqref{eq:RH-condition}, described by shock fronts. 
Firstly we consider a system in the disjointed fashion where both left state, $\textbf{u}_\leftstate=[\mathcal{C}_\leftstate,\mathcal{H}_\leftstate]^T$, and right state, $\textbf{u}_\rightstate=[\mathcal{C}_\rightstate,\mathcal{H}_\rightstate]^T$, reside in the same region. Then we investigate more complicated cases where left and right states can lie in different regions. Lastly we discuss the cases where a fully-saturated region forms, leading to the failure of the current hyperbolic PDE analysis. This theory is then further extended to analyze the formation and evolution of fully-saturated regions which are governed by a different, elliptic PDE \citep{shadab2022analysis}. {Although there can be at most one moving wave for simple cases where the medium remains unsaturated ($C<\rho$, $\mathcal{C}<1$ or $\phi_g>0$), there can be two moving waves when a fully-saturated region appears, i.e., $\mathcal{C}=1$.}

\subsubsection{\textbf{Region 2 only (Three-phase region)}}
Region 2 ($0<\mathcal{H}<\mathcal{C}$) consists of all three-phases, which is relevant for temperate glaciers where $\mathcal{T}=0$ or $T=T_m$. In three-phase region, the eigenvalues of the flux gradient \eqref{eq:dimless-Jacobian} are 

    \begin{align}\label{eq:Evals-region2}
        \lambda_1 = 0 \quad \textrm{and} \quad \lambda_2 = f_{,\mathcal{C}} + f_{,\mathcal{H}}=n\mathcal{H}^{n-1} (1-\mathcal{C}+\mathcal{H})^{m-n}.
    \end{align}
 where the subscripts ${\mathcal{C}}$ and ${\mathcal{H}}$ refer to the partial derivatives with respect to the corresponding conserved variable.

\begin{figure}
\centering 
  \includegraphics[width=\linewidth]{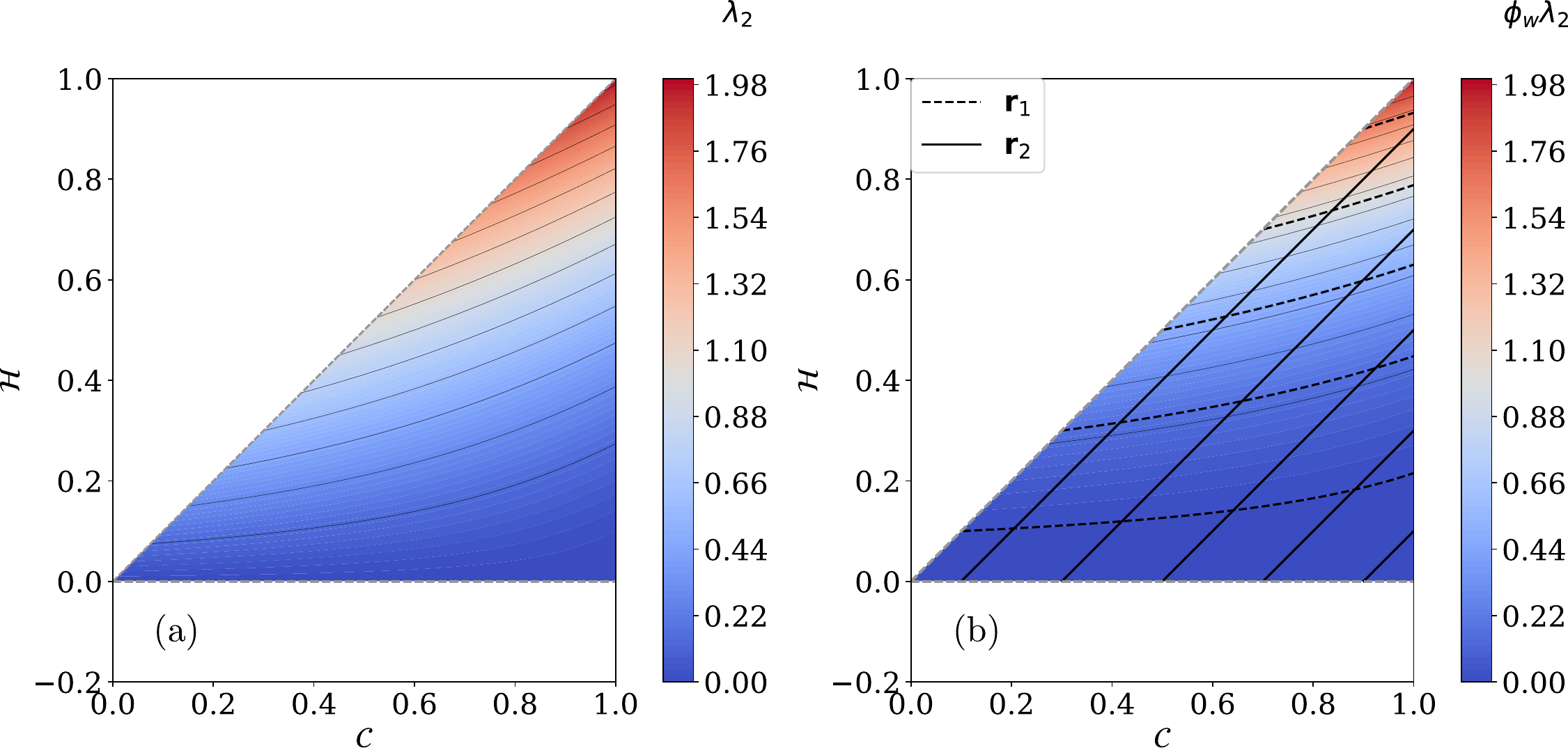}
\caption{Contour plots of (a) second eigenvalue, $\lambda_2$, and (b) propagation velocity of second front with respect to the melt given by $\phi_w\lambda_2$ in hodograph plane ($\mathcal{C}\mathcal{H}$ phase space) for $m=3$ and $n=2$. The slow path $\textbf{r}_1$ and fast path $\textbf{r}_2$ are shown with dashed and solid lines respectively in panel \textit{b}.}
\label{fig:3}
\end{figure}

 The first reaction front is a stationary contact discontinuity as $\lambda_1=0$. The eigenvalue $\lambda_2$ gives a dimensionless propagation speed of the second reaction front $\mathscr{W}_2$ as a function of $\mathcal{C}$ and $\mathcal{H}$, as plotted in Figure~\ref{fig:3}\emph{a} for $m=3$ and $n=2$. Therefore, all  solutions governed by hyperbolic PDEs \eqref{eq:dimless-system-gov-eqs-1D} will have at most a single moving reaction front. Due to variable porosity, the dimensionless system is scaled with respect to the largest saturated hydraulic conductivity, $K_h$, which corresponds to the volumetric flux of water (Darcy's flux), rather than the melt velocity. Hence, the propagation speed of the second reaction front relative to the melt is $\phi_w \lambda_2$, as shown in Figure \ref{fig:3}\emph{b}. As expected, the propagation velocity $\phi_w \lambda_2$ has significantly retarded near the solidus, as shown by more separated contours. These eigenvalues yield two corresponding, linearly independent eigenvectors in $\mathcal{C}\mathcal{H}$ hodograph plane given by 

    \begin{align}
        \textbf{r}_1 = \begin{bmatrix} \frac{-f_{,\mathcal{H}} }{f_{,\mathcal{C}} } \\ 1 \end{bmatrix} = \begin{bmatrix} \frac{(n (1-\mathcal{C})+m\mathcal{H})}{(m-n)\mathcal{H}} \\ 1 \end{bmatrix} \quad \textrm{and} \quad
        \textbf{r}_2 = \begin{bmatrix} \frac{f_{,\mathcal{H}} }{f_{,\mathcal{H}} } \\ 1 \end{bmatrix} = \begin{bmatrix} 1 \\ 1 \end{bmatrix}.
    \end{align}
These eigenvectors can be used to find the integral curves using the system of ODEs

\begin{align}
\frac{\d \textbf{u}}{\d \eta} &=  \frac{1}{\nabla_{\textbf{u}} \lambda_p \cdot \textbf{r}_p} \textbf{r}_p \label{eq:ODE-system}  \end{align}
which can be further integrated to obtain the solution pathways $\textbf{u}(\textbf{u}_0,\eta) $ as

\begin{align}
\textbf{u}(\textbf{u}_0,\eta) &= \textbf{u}_0 + \int_{\lambda_p(\textbf{u}_0)}^\eta \frac{1}{\nabla_{\textbf{u}} \lambda_p \cdot \textbf{r}_p} \textbf{r}_p ~\d \eta' \label{eq:ODE-system-solution}.
    \end{align}
    
These paths in $\mathcal{C}\mathcal{H}$ hodograph plane comprise the set of states that can be connected to the initial state $\textbf{u}_0$ by a reaction front with a continuous variation in \textbf{u}. In three phase region, the family of integral curves corresponding to first eigenvector $\textbf{r}_1$ is referred to as slow path as $\lambda_1 <\lambda_2$ and is given by

   \begin{align} 
        \mathcal{C}&=1+ \mathcal{H}+\mathfrak{C} \mathcal{H}^{-\frac{n}{m-n}} \label{eq:slow-path}.
    \end{align}
    where $\mathfrak{C}$ is the constant of integration, which can be found for the initial point $\textbf{u}_0$. The slow path lines are the same as constant flux lines, as shown in Figure \ref{fig:3}\textit{b}. Next, the family of integral curves corresponding to the second eigenvector is known as fast path and is given by
    
    \begin{align}
        \mathcal{C} &= \mathcal{H} + \mathfrak{C} .\label{eq:fast-path}
    \end{align}
 The speed of second characteristic $\lambda_2$ is non-negative and increases monotonically in the direction of integral curves corresponding to the second eigenvector $\textbf{r}_2$ (fast paths) in three-phase region for $n>1$ (see Lemma Appendix 1.1 and Figure \ref{fig:3}\emph{b}). Additionally, the slow path corresponds to constant flux contours as $\lambda_1=0$ and the fast path corresponds to constant porosity, $\varphi$, contours (Lemma Appendix 1.2).

\subsubsection*{Solutions in three-phase region}
We will now discuss the different scenarios involving three-phase region only. The discussion below assumes that $\textbf{u}_0$ is the left state, $\textbf{u}_\leftstate$, and describes the set of permissible right states $\textbf{u}_\rightstate$.

\paragraph{{(a.) Stationary linear reaction front (Case I):}} The integral curves associated with $\lambda_1$ and $\textbf{r}_1$, known as the first characteristic field ($\lambda_1,\textbf{r}_1$), are constant flux lines in the $\mathcal{C}\mathcal{H}$ phase space (or hodograph plane). Any right state $\textbf{u}_\rightstate$, along the integral curve, connected to $\textbf{u}_0$ by a stationary discontinuity is a weak solution of Equation \eqref{eq:dimless-system-linearized}. Because $\lambda_1=0$, the first wave is a stationary contact discontinuity $\mathscr{C}_1$. The fluxes of $\mathcal{C}$ and $\mathcal{H}$ on both sides are the same so that the melt transport does not change $\mathcal{C}$ and $\mathcal{H}$ and the front does not evolve. So, a contact discontinuity is the solution for the left and right states lying on the slow path \eqref{eq:slow-path} (constant flux lines) satisfying

\begin{align} \label{eq:stationary-linear-rxn-front}
    \frac{\varphi_\leftstate}{\varphi_\rightstate}&= \frac{1 - \mathcal{C}_\leftstate +\mathcal{H}_\leftstate}{1 - \mathcal{C}_\rightstate + \mathcal{H}_\rightstate} = \left(\frac{\mathcal{H}_\leftstate}{\mathcal{H}_\rightstate}\right)^{-\frac{n}{m-n}} \quad \textrm{or} \quad f(\textbf{u}_l)=f(\textbf{u}_r).
\end{align}

The final solution in this case takes the form 

\begin{align}\label{eq:contact-disc-only}
    \textbf{u} &= \begin{cases}\textbf{u}_l, \quad & \zeta < 0 \\  \textbf{u}_r, \quad & \zeta>0\end{cases}.
\end{align} 
Figures \ref{fig:contact_shock_rarefaction}\textit{a}-\textit{c} (green color) illustrate an example of a system that results in a stationary contact discontinuity. This system corresponds to a steady meltwater flux of $f=0.112$ inside temperate porous firn with a jump in porosity from 70\% at $\zeta<0$ to 55.3\% in $\zeta>0$ (coarse-to-fine transition in firn) leading to liquid water contents of $0.4$ and $0.45$ respectively in these regions. These values correspond to  $\textbf{u}_l=[\mathcal{C}_l,\mathcal{H}_l]^T=[0.7,0.4]^T$ and right state $\textbf{u}_r=[0.897,0.45]^T$. The resulting system illustrated by the volume fractions of the three phases highlights a porosity jump in the temperate firn that has a constant steady meltwater flux on both sides, as shown in Figures \ref{fig:contact_shock_rarefaction}\textit{d}-\textit{f} at different times. This case is summarized in Table \ref{table:2}.

\begin{figure}
    \centering
    \hspace*{-2cm}
    \includegraphics[width=1.2\linewidth]{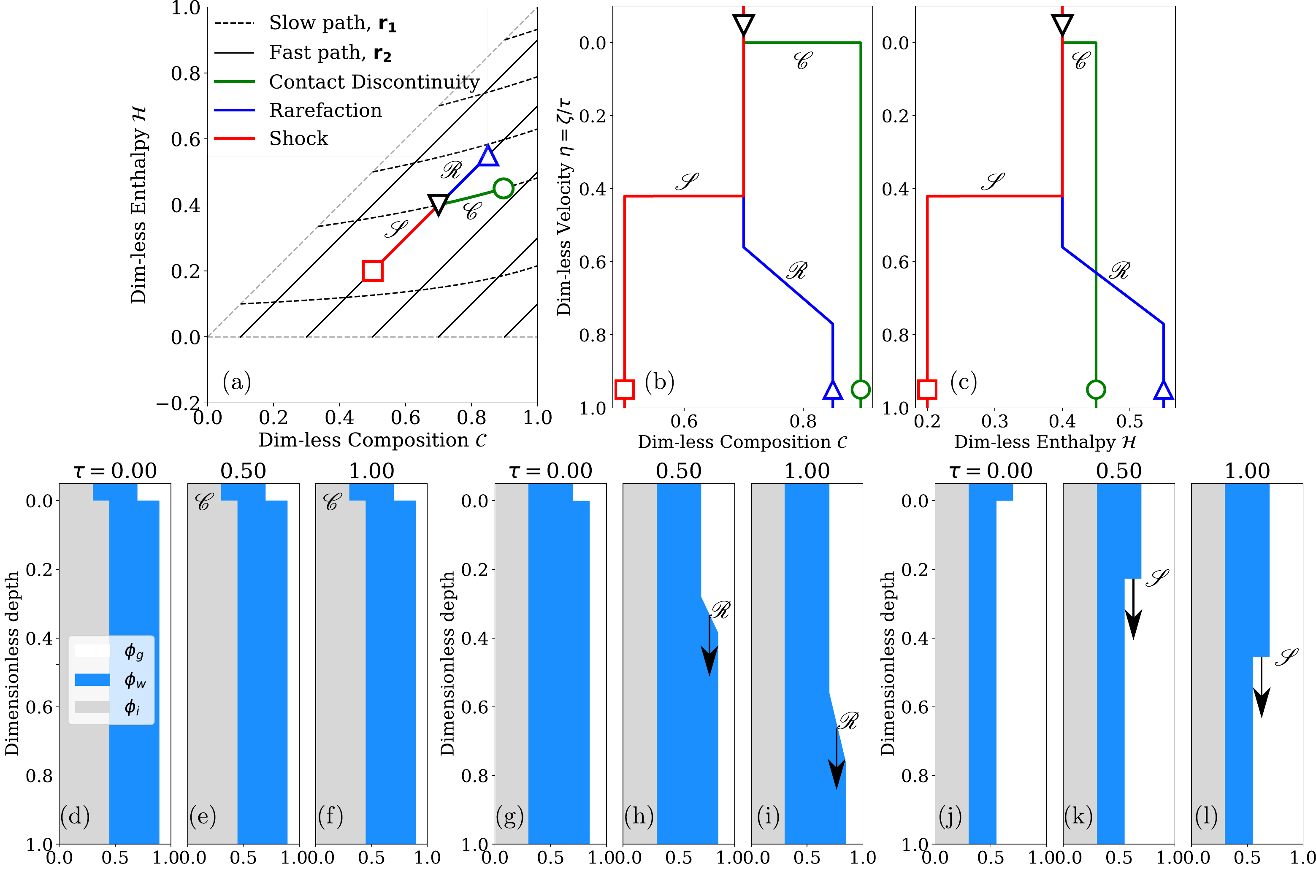}
    \caption{The simple solutions of Riemann problem leading to a contact discontinuity $\mathscr{C}$ (green), rarefaction $\mathscr{R}$ (blue), and shock waves $\mathscr{S}$ (red). (a) Construction of solution in hodograph plane and their corresponding self-similar analytical solutions for (b) dimensionless composition and (c) dimensionless enthalpy with dimensionless velocity $\eta$. The evolution of the volume fractions of the three phases in the system for different configurations at dimensionless times $\tau=0,0.5,1.0$: (d-f) Case I - contact discontinuity, (g-i) Case II - drying front/rarefaction wave, and (j-l) Case III - wetting front/shock wave.}
    \label{fig:contact_shock_rarefaction}
\end{figure}

\paragraph{{(b.) Moving nonlinear reaction front:}} The integral curves associated with the second characteristic field ($\lambda_2,\textbf{r}_2$) are the constant porosity $\varphi$ contours. Thus, a nonlinear characteristic wave is the solution for the left and right states lying on the fast path \eqref{eq:fast-path} satisfying

\begin{align}\label{eq:47}
    \mathcal{C}_\rightstate&=\mathcal{C}_\leftstate + (
    \mathcal{H}_\rightstate-\mathcal{H}_\leftstate).
\end{align} 

\subparagraph{(i.) Rarefaction wave (Case II):} If the characteristic speed $\lambda_2$ varies smoothly from left to right state, any right state $\textbf{u}$ along the fast path is connected to left state $\textbf{u}_0$ by a continuously varying reaction front (Figure \ref{fig:contact_shock_rarefaction}\textit{a}, blue lines). The propagation velocity, $\lambda_2$, along these continuous reaction fronts increases monotonically such that the reaction front spreads with time. These self-smoothening drying fronts are referred to as rarefaction waves, denoted by the symbol $\mathscr{R}$. The term drying front refers to drying due to meltwater drainage instead of refreezing \citep{clark2017analytical}. Rarefaction waves are a weak solution of Equation \eqref{eq:37} if the resultant profile of $\textbf{u}$ is single-valued. This condition is satisfied if $\textbf{u}$ lies on the branch of the integral curve $\textbf{r}_2$ emanating from $\textbf{u}_0$ in the direction of increasing $\lambda_2$ (see Figures \ref{fig:3} and \ref{fig:contact_shock_rarefaction}\textit{a}, blue lines). The analytical solution concerning the self-similar variable $\eta$ for a rarefaction wave on an integral curve can be evaluated from Equation \eqref{eq:ODE-system-solution}, which comes out to be

\begin{align} 
     \mathcal{H}&= \sqrt[n-1]{\frac{\eta}{n (1-\mathcal{C}_\leftstate + \mathcal{H}_\leftstate)^{m-n}}}\quad \text{and} \label{eq:rarefaction-soln-C}\\
     \mathcal{C} &= \mathcal{C}_\leftstate -\mathcal{H}_\leftstate + \sqrt[n-1]{\frac{\eta}{n (1-\mathcal{C}_\leftstate + \mathcal{H}_\leftstate)^{m-n}}}. \label{eq:rarefaction-soln-H}
\end{align}
The final solution in this case takes the form
\begin{align}\label{eq:final-sol-rarefaction-only}
    \textbf{u} &= \begin{cases}\textbf{u}_l, \quad & \eta < \lambda_2(\textbf{u}_l) \\  
 \begin{bmatrix}\mathcal{C}_\leftstate -\mathcal{H}_\leftstate + \sqrt[n-1]{\frac{\eta}{n (1-\mathcal{C}_\leftstate + \mathcal{H}_\leftstate)^{m-n}}}\\ \sqrt[n-1]{\frac{\eta}{n (1-\mathcal{C}_\leftstate + \mathcal{H}_\leftstate)^{m-n}}} \end{bmatrix}, \quad & \lambda_2(\textbf{u}_l) <\eta<\lambda_2(\textbf{u}_r)
    \\\textbf{u}_r, \quad & \eta>\lambda_2(\textbf{u}_r)\end{cases}.
\end{align}
where the speed of the second characteristic $\lambda_2(\cdot)$ is evaluated from Equation \eqref{eq:Evals-region2}.
An example of this case is when meltwater flux instantly drops on a glacier, leading to a smoothening drainage front. Figures \ref{fig:contact_shock_rarefaction}\textit{a}-\textit{c} (blue line) show a moving rarefaction developed inside a 70\% porous firn due to an instantaneous reduction in meltwater flux, captured by lower (40\%) and higher (55\%) liquid water content layers on top and bottom respectively. These numbers translate to left (top) $\textbf{u}_l=[0.7, 0.4]^T$ and right (bottom) states being $\textbf{u}_r=[0.85,0.55]^T$. The resulting evolution of volume fractions is shown in Figures \ref{fig:contact_shock_rarefaction}\textit{g}-\textit{i}. The leading edge of the rarefaction front moves faster than the trailing edge connected by a gradual, linear smoothening. In this case, it is a linear profile (straight line) as the power law exponents in Equation \eqref{eq:final-sol-rarefaction-only} are $m=3$ and $n=2$. Drying/rarefaction front has been studied by \cite{clark2017analytical} and was observed in models studying the Dye-2 site in Greenland on 12 August 2016 after the meltwater flux ceased \citep{samimi2020meltwater,vandecrux2020firn,colliander2022ice}.

\subparagraph{(ii.) Shock wave (Case III):} If the right state $\textbf{u}$ lies on the opposite branch of the integral curve (shown by red line in Figure \ref{fig:contact_shock_rarefaction}\textit{a} for example), a continuous reaction front would result in unphysical solutions as the characteristics will cross each other. Therefore, in this case $\textbf{u}$ is connected to left state $\textbf{u}_0$ by a discontinuous reaction front that propagates with a velocity, $\Lambda_{\mathscr{S}}(\textbf{u}_0,\textbf{u})$, which can be calculated from R-H jump condition \eqref{eq:RH-condition} using initial conditions ($\textbf{u}_l,\textbf{u}_r$). Such fronts are referred to as wetting/shock fronts, denoted by the symbol $\mathscr{S}$. The set of permissible right states $\textbf{u}$ that can be connected to the left state $\textbf{u}_0$ by shocks lie on the segment of the Hugoniot-locus that satisfies the entropy condition. In the three-phase region for the system of equations considered, the Hugoniot-locus is the same as the integral curve, which is found to be the fast path $\textbf{r}_2$ \eqref{eq:fast-path} from the Hugoniot jump condition \eqref{eq:RH-condition}, since the flux of enthalpy and composition are the same. The dimensionless velocity of the shock from Equation \eqref{eq:RH-condition} is then 

\begin{align}\label{eq:RH-condition-shock-case}
\Lambda_{\mathscr{S}}(\textbf{u}_+,\textbf{u}_-) = \frac{\d \zeta}{\d \tau} = \frac{f(\textbf{u}_+)-f(\textbf{u}_-)}{\mathcal{H}_+-\mathcal{H}_{-}} = \frac{f(\textbf{u}_+)-f(\textbf{u}_-)}{\mathcal{C}_+-\mathcal{C}_{-}},
\end{align}

where the symbol $-$ is the left (top) state and $+$ is the right (bottom) state for this particular configuration where $f(\textbf{u}_l)>f(\textbf{u}_r)$. The final solution thus takes the form 

\begin{align}\label{eq:final-sol-shock-only}
    \textbf{u} &= \begin{cases}\textbf{u}_l, \quad & \zeta /  \tau < \Lambda_{\mathscr{S}}(\textbf{u}_r,\textbf{u}_l)
    \\ \textbf{u}_r, \quad & \zeta /  \tau>\Lambda_{\mathscr{S}}(\textbf{u}_r,\textbf{u}_l)  \end{cases}.
\end{align}
An example of this case is when meltwater flux instantly increases on a temperate glacier, it leads to a sharp wetting front propagating in the direction of gravity. Figures \ref{fig:contact_shock_rarefaction}\textit{j}-\textit{k} show a moving shock front developed inside a 70\% porous firn due to an instantaneous increase in meltwater flux, captured by higher (40\%) and lower (20\%) liquid water content $\phi_w$ in top and bottom layers respectively. These numbers translate to left (top) state being $\textbf{u}_l=[0.7, 0.4]^T$ and right (bottom) state being $\textbf{u}_r=[0.5,0.2]^T$ sketched analytically in Figures \ref{fig:contact_shock_rarefaction}\textit{a}-\textit{c}. The wetting front has been discussed \citep[e.g.]{colbeck1972theory,humphrey2012thermal,meyer2017continuum,clark2017analytical} and observed in models \citep[e.g.]{vandecrux2020firn, samimi2021time, colliander2022ice} and field observations including at the Dye-2 site in Greenland on 9 August 2016 when meltwater percolates in temperate firn \citep{heilig2018seasonal,samimi2020meltwater}.

\paragraph{{{(c.) Two fronts with an intermediate state:}}}
The solution profile contains a single reaction front if $\textbf{u}_\leftstate$ and $\textbf{u}_\rightstate$ share the same integral curve or Hugoniot-locus. In all other cases in the three-phase region without complete saturation, a different state than left or right state forms which is referred to as an intermediate state, $\textbf{u}_i$, in the hodograph plane. At this intermediate state, the solution switches from the first characteristic field $(\lambda_1,\textbf{r}_1)$ to the second $(\lambda_2,\textbf{r}_2)$. In other words, at the intermediate state $\textbf{u}_i$, the solution changes from the stationary front along constant flux lines (slow path) to the advancing reaction front along the path of constant porosity contours (fast path) that also mimic the liquidus ($\mathcal{H}=\mathcal{C}$).
The two possible intermediate states are given by the intersections of the integral curves emanating from $\textbf{u}_\leftstate$ and $\textbf{u}_\rightstate$ (see Figure \ref{fig:two-intermediate-states}\textit{a} for example). Only one intersection yields a physically realistic single-value solution. The correct intersection is selected by requiring that the propagation speed increases monotonically from $\textbf{u}_\leftstate$ and $\textbf{u}_\rightstate$. A single-valued solution is ensured \textit{if and only if} $\textbf{u}_\leftstate$ and $\textbf{u}_i$ lie on the slow path first and then $\textbf{u}_i$ is connected to $\textbf{u}_\rightstate$ along the fast path. The reactive melt transport system considered here only allows two solutions for this case:
\begin{align}
\textbf{u}_\leftstate\xrightarrow{\mathscr{C}_1}\textbf{u}_i \xrightarrow{\mathscr{R}_2}\textbf{u}_\rightstate \quad \textrm{and} \quad \textbf{u}_\leftstate\xrightarrow{\mathscr{C}_1}\textbf{u}_i \xrightarrow{\mathscr{S}_2}\textbf{u}_\rightstate,
\end{align}
because the first characteristic is linearly degenerate and the reaction front along the slow path is always a contact discontinuity $\mathscr{C}_1$. The reactive melt transport across an initial discontinuity is characterized by the formation of a reacted zone corresponding to $\textbf{u}_i$ that is bounded between a stationary front $\mathscr{C}_1$, and an advancing front that is either a rarefaction wave $\mathscr{R}_2$ or a shock wave $\mathscr{S}_2$. Below we will discuss these two cases.

\begin{figure}
    \centering
    \includegraphics[width=1.0\linewidth]{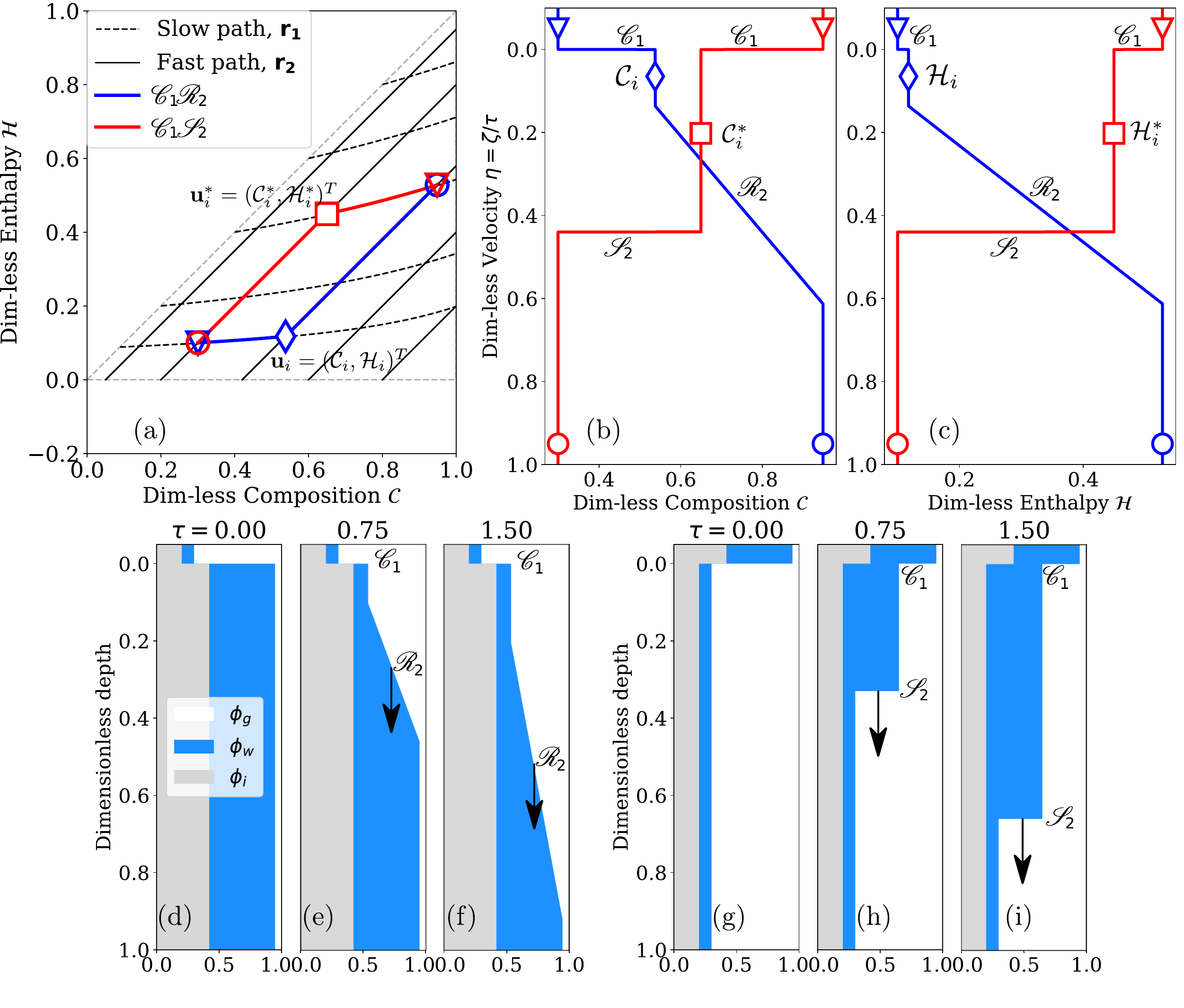}
    \caption{Formation of an intermediate state $\textbf{u}_i$ or $\textbf{u}^*_i$ for Case IV - $\mathscr{C}_1\mathscr{R}_2$ or Case V - $\mathscr{C}_1\mathscr{S}_2$ respectively. An asterisk is used to differentiate the two intermediate states corresponding to the two cases. (a) Construction of solution in hodograph plane and their corresponding self-similar analytical solutions for (b) dimensionless composition and (c) dimensionless enthalpy with dimensionless velocity $\eta$. The evolution of the volume fractions of the three phases in the system at dimensionless times $\tau=0,0.5,1.0$ for the two configurations: (d-f) Case IV - $\mathscr{C}_1\mathscr{R}_2$ and (g-i) Case V - $\mathscr{C}_1\mathscr{S}_2$.}
    \label{fig:two-intermediate-states}
\end{figure}

\subparagraph{(i.) 1-Contact discontinuity and 2-Rarefaction (Case IV):} In this case $f(\textbf{u}_l)<f(\textbf{u}_r)$ and the resulting first characteristic wave is a contact discontinuity, $\mathscr{C}_1$, which satisfies Equation \eqref{eq:slow-path} for left state $\textbf{u}_\leftstate$ and the intermediate state $\textbf{u}_i$. The second characteristic wave is a rarefaction which is governed by Equations \eqref{eq:fast-path}, \eqref{eq:rarefaction-soln-C} and \eqref{eq:rarefaction-soln-H} for intermediate state $\textbf{u}_i$ and right state $\textbf{u}_r$. Combining all these equations results in a nonlinear algebraic equation to evaluate $\mathcal{H}_i$ that is given by

\begin{align} \label{eq:52}
    {1 +\mathcal{H}_\rightstate- \mathcal{C}_\rightstate-\mathcal{H}_i^{-\frac{1}{n-1}} {\mathcal{H}_\leftstate}^{\frac{n}{n-1}} { \left(\frac{1-\mathcal{C}_\leftstate + \mathcal{H}_\leftstate}{1-\mathcal{C}_\rightstate + \mathcal{H}_\rightstate}\right)^{\frac{m-n}{n-1}}}+ \mathcal{H}_i}&=\left({1+ \mathcal{H}_\leftstate-\mathcal{C}_\leftstate}\right) \left(\frac{\mathcal{H}_i}{\mathcal{H}_\leftstate}\right)^{-\frac{n}{m-n}}
\end{align}
which can be re-written in terms of porosities, $\varphi$, as

\begin{align}\label{eq:1.169new}
    {\varphi_\rightstate-\mathcal{H}_i^{-\frac{1}{n-1}} {\mathcal{H}_\leftstate}^{\frac{n}{n-1}} { \left(\frac{\varphi_\leftstate}{\varphi_\rightstate}\right)^{\frac{m-n}{n-1}}}+ \mathcal{H}_i}&={\varphi_\leftstate} \left(\frac{\mathcal{H}_i}{\mathcal{H}_\leftstate}\right)^{-\frac{n}{m-n}}.
\end{align}
Next, the composition at intermediate state, $\mathcal{C}_i$, can be computed from Equations \eqref{eq:fast-path} and \eqref{eq:Evals-region2} which corresponds to the fast path and the speed of the second characteristic, $\lambda_2$, respectively. The final solution in this case takes the form

\begin{align}\label{eq:final-sol-contact-rarefaction}
    \textbf{u} &= \begin{cases}\textbf{u}_l, \quad &  \zeta < 0 \\  \textbf{u}_i, \quad &  0< \zeta / \tau < \lambda_2(\textbf{u}_i) \\  
 \begin{bmatrix}\mathcal{C}_i -\mathcal{H}_i + \sqrt[n-1]{\frac{\eta}{n (1-\mathcal{C}_i + \mathcal{H}_i)^{m-n}}}\\ \sqrt[n-1]{\frac{\eta}{n (1-\mathcal{C}_i + \mathcal{H}_i)^{m-n}}} \end{bmatrix}, \quad & \lambda_2(\textbf{u}_i) <\eta<\lambda_2(\textbf{u}_r)
    \\\textbf{u}_r, \quad & \eta>\lambda_2(\textbf{u}_r)\end{cases}.
\end{align}

{An example of this case is a sudden drop in meltwater flux inside temperate firn where porosity also reduces with depth (Figure \ref{fig:two-intermediate-states}\textit{d}). Physically, the first contact discontinuity represents the constant flux of water entering the bottom, low porosity layer (Figures \ref{fig:two-intermediate-states}\textit{e},\textit{f}). The second rarefaction shows the drainage of the wetter firn due to gravity. Figure~\ref{fig:two-intermediate-states}\emph{a} shows the construction of the solution $\mathscr{C}_1 \mathscr{R}_2$ in the blue line for the left (top) state which is more porous $\varphi_l=80\%$, and has less water content (LWC) $\phi_{w,l}=0.1$. The right (bottom) state is less porous $\varphi_r=58.0\%$ but has more liquid water content $\phi_{w,r}=0.528$. These values correspond to left and right states being $\textbf{u}_l=(0.3,0.1)^T$ and $\textbf{u}_r=(0.948,0.528)^T$ respectively shown in Figures \ref{fig:two-intermediate-states}\textit{a}-\textit{c}. The intermediate state comes out to be $\textbf{u}_i=(0.538,0.117)^T$ which corresponds to LWC $\phi_{w,i}=11.7\%$ and the same porosity as the right state, i.e., $\varphi_i=58.0\%$. Figures~\ref{fig:two-intermediate-states}\emph{b} and \ref{fig:two-intermediate-states}\emph{c} show the corresponding self-similar analytical solutions for composition and enthalpy respectively for this case with blue lines which only depend on the dimensionless velocity. The rarefaction moves down with a characteristic velocity that can be computed analytically. Figures~\ref{fig:two-intermediate-states}\emph{d}-\emph{f} shows the resulting evolution of volume fraction of each phase at different times showing self-similar expansion of the rarefaction wave.

\subparagraph{(ii.) 1-Contact discontinuity and 2-Shock (Case V):}
This case is similar to Case IV but has $f(\textbf{u}_l)>f(\textbf{u}_r)$. As a result, the first characteristic wave is a contact discontinuity, $\mathscr{C}_1$, that satisfies Equation \eqref{eq:slow-path} for the left state $\textbf{u}_\leftstate$ and intermediate state $\textbf{u}_i$. Since the second characteristic lies on the Hugoniot locus, the result is a shockwave, $\mathscr{S}_2$, which satisfies the Hugoniot-jump condition \eqref{eq:RH-condition} for the intermediate state and right state. Combining Equations \eqref{eq:slow-path} and \eqref{eq:fast-path} gives a simple relation for dimensionless enthalpy at intermediate state, $\mathcal{H}_i$, to be 

\begin{align}
      \mathcal{H}_i&=\mathcal{H}_\leftstate {\left(\frac{\varphi_\leftstate}{\varphi_\rightstate}\right)}^{\frac{m-n}{n}} = \mathcal{H}_\leftstate {\left(\frac{1+ \mathcal{H}_\leftstate-\mathcal{C}_\leftstate}{1+ \mathcal{H}_\rightstate-\mathcal{C}_\rightstate}\right)}^{\frac{m-n}{n}}.\label{eq:intermediate-Hi-temperate-C1R2}
\end{align}

Then Equation \eqref{eq:fast-path} for the intermediate state, $\textbf{u}_i$, and right state, $\textbf{u}_\rightstate$, provides the value of dimensionless composition at intermediate state, $\mathcal{C}_i$. The final solution in this case takes the form

\begin{align}\label{eq:final-sol-contact-shock}
    \textbf{u} &= \begin{cases}\textbf{u}_l, \quad &  \zeta < 0 \\  \textbf{u}_i, \quad &  0 < \zeta / \tau < \Lambda_\mathscr{S}(\textbf{u}_i,\textbf{u}_r) \\  \textbf{u}_r, \quad & \zeta / \tau>\Lambda_\mathscr{S}(\textbf{u}_i,\textbf{u}_r)\end{cases}.
\end{align}

{An example of this case is a sudden rise in meltwater flux inside temperate firn where porosity also reduces with depth (Figure~\ref{fig:two-intermediate-states}\textit{g}). Physically, the first contact discontinuity represents the increased, constant flux of water entering the bottom, low porosity layer. The second shock shows the wetting front advancing the water content to dryer firn because of gravity (Figures~\ref{fig:two-intermediate-states}\textit{h},\textit{i}). Figure~\ref{fig:two-intermediate-states}\emph{a} shows the construction of the solution $\mathscr{C}_1 \mathscr{S}_2$ with red line for the left state which is less porous, i.e., $\varphi_l=58.0\%$, but has more water content ($\phi_{w,l}=52.8\%$), and the right state is more porous and less wet corresponding to $\varphi_r = 80\%$ and LWC $\phi_{w,r}=10\%$. These values correspond to $\textbf{u}_l=(0.948,0.528)^T$ and $\textbf{u}_r=(0.3,0.1)^T$ in Figures \ref{fig:two-intermediate-states}\textit{a}-\textit{c}. The intermediate state comes out to be $\textbf{u}^*_i=(0.648,0.448)^T$ which corresponds to the wetter intermediate region behind the wetting front with LWC $\phi_{w,i}=44.8\%$ but the same porosity as the right state, i.e., $\varphi_i=80\%$ (see Figures~\ref{fig:two-intermediate-states}\emph{g}-\emph{i}).} Figures~\ref{fig:two-intermediate-states}\emph{b} and \ref{fig:two-intermediate-states}\emph{c} show the corresponding self-similar analytical solutions for composition and enthalpy respectively for this case with red lines which only depend on the dimensionless velocity.

\paragraph{{{(d.) Two fronts and a jump with two intermediate states (formation of a saturated region, Case VI):}}} The initial conditions of this case are similar to Case V ($\mathscr{C}_1\mathscr{S}_2$). However, in this case, the slow path emanating from the left state does not intersect with the fast path from the right state in the three-phase region where $\mathcal{C}<1$. In other words, the bottom layer is unable to accommodate the flux of meltwater from the top layer, as discussed for soils in \cite{shadab2022analysis}. If the intermediate state leads to complete saturation, i.e., $\mathcal{C}_i=1$, then the proposed hyperbolic PDE solution framework breaks down. This happens because inside the fully-saturated region, the dynamics of the water phase change from gravity-driven to pressure-driven as the governing model changes from hyperbolic (local) to elliptic (global) partial differential equation (see \cite{shadab2022analysis} for a detailed analysis). In this case, the dynamics becomes complicated as the complete solution cannot be directly interpreted from the hodograph plane because the flux in the saturated region may not simply be the hydraulic conductivity \eqref{eq:dimless-flux-CH-vector} anymore.

Nevertheless, we can still construct a full analytical solution to this problem using the extended kinematic wave approximation proposed by \cite{shadab2022analysis}. In this case, the solution consists of three waves including a backfilling shock moving upwards, denoted by symbol $\mathscr{S}_1^*$, a stationary jump at the initial location of the jump, denoted by $\mathscr{J}_2$, and a downward moving shock (wetting front), $\mathscr{S}_3$, into the less porous and temperate layer. Note that the jump $\mathscr{J}_2$ lies in the saturated region and therefore does not represent any hyperbolic wave. Therefore, it is represented by a broken arrow in the full solution given by

\begin{align*}
\textbf{u}_\leftstate\xrightarrow{\mathscr{S}^*_1}\textbf{u}_{i_1}\xdashrightarrow{\mathscr{J}_2}\textbf{u}_{i_2} \xrightarrow{\mathscr{S}_3}\textbf{u}_\rightstate.
\end{align*}

The solution is explained in detail as follows. First, the backfilling shock on the fast path (constant porosity line) connects to the first intermediate state, $i_1$, which lies on the line $\mathcal{C}=1$. Therefore, the first intermediate state variables are

\begin{align}\label{eq:i1_variables}
    \mathcal{C}_{i_1} = 1 \quad \text{and} \quad \mathcal{H}_{i_1}=1-\mathcal{C}_\leftstate+\mathcal{H}_\leftstate.
\end{align}

This state is observed right next to the left state and can be considered as the rising perched water table in the region $\zeta < 0 $. The speed of this backfilling, upper shock in dimensionless form is again given by Rankine-Hugoniot jump condition \citep{leveque1992numerical} as

\begin{align}\label{eqn:Su}
    \Lambda_{\mathscr{S}^*_1} = \frac{\d \zeta_U}{\d \tau} = \frac{f(\textbf{u}_\leftstate)-q_s(\tau)}{\mathcal{H}_\leftstate-\mathcal{H}_{i_1}}= \frac{f(\textbf{u}_\leftstate)-q_s(\tau)}{\mathcal{C}_\leftstate-\mathcal{C}_{i_1}} < 0,
\end{align}    
where $\zeta_U$ is the location of the upper shock. The shock moves upwards due to choking as the numerator is positive since flux $f(\textbf{u}_\leftstate)$ is more than the time-dependent dimensionless flux in the saturated region, $q_s(\tau)$, which is also scaled by $K_h$. Note that the flux in the saturated region $q_s(\tau)$ is not the saturated hydraulic conductivity but instead, it is governed by the dynamics of the saturated region.

Next, the first intermediate state $i_1$ is connected to the second intermediate state $i_2$ through a stationary jump $\mathscr{J}_2$ at the location of initial jump at $\zeta = 0$. Both intermediate states lie in the fully-saturated region and therefore the jump $\mathscr{J}_2$ does not represent a hyperbolic wave. The flux inside the isothermal saturated region $q_s(\tau)$ is found to be uniform in this case \citep{shadab2022analysis}. Therefore, the flux between the two intermediate states is also same, equal to $q_s(\tau)$. The state variables $\mathcal{C}$ and $\mathcal{H}$ for the second intermediate state, $i_2$, are provided by the right state. The second intermediate state also lies at $\mathcal{C}=1$ on the fast path \eqref{eq:fast-path} (Hugoniot locus) emanating from right state, i.e., $\mathcal{C}_r-\mathcal{H}_r=\mathcal{C}_{i_2}-\mathcal{H}_{i_2}$. Therefore, the second intermediate state variables are simply,

\begin{align}\label{eq:i2_variables}
    \mathcal{C}_{i_2} = 1 \quad \text{and} \quad \mathcal{H}_{i_2} = 1 - \mathcal{C}_\rightstate + \mathcal{H}_\rightstate . 
\end{align}

Similarly, the velocity of the downward-moving lower shock (wetting front) is

\begin{align}\label{eqn:Sl}
    \Lambda_{\mathscr{S}_3} = \frac{\d \zeta_L}{\d \tau} = \frac{q_s(\tau)-{f{(\textbf{u}_\rightstate})}}{\mathcal{C}_{i_2}-\mathcal{C}_\rightstate} = \frac{q_s(\tau)-f{(\textbf{u}_\rightstate})}{1-\mathcal{C}_\rightstate} > 0 
\end{align}
where $\zeta_L$ is the dimensionless location of the lower shock. Similar to a two-layer soil \cite{shadab2022analysis}, the dimensionless flux in the saturated region is a depth-based harmonic mean of the dimensionless saturated hydraulic conductivities at the two intermediate states given by

\begin{align}\label{eq:harmonic-two-layer-final}
   q_s(\tau) =\frac{\zeta_U(\tau) - \zeta_L(\tau)}{\frac{\zeta_U(\tau)}{K_{i_1}}-\frac{\zeta_L(\tau)}{K_{i_2}}},
\end{align}
where $K_{i_1}$ and $K_{i_2}$ are the dimensionless saturated hydraulic conductivities at the first and second intermediate states given by
\begin{align}\label{eq:sat-hyd-cond-sat-regions}
    K_{i_1} = \varphi_{i_1}^m = (1-\mathcal{C}_\leftstate+\mathcal{H}_\leftstate)^m \quad \text{and} \quad K_{i_2} = \varphi_{i_2}^m = (1-\mathcal{C}_\rightstate+\mathcal{H}_\rightstate)^m
\end{align}
using Equations \eqref{eq:i1_variables} and \eqref{eq:i2_variables}. Note that the porosities at the first and second intermediate states are the porosities of the left and right state respectively. Solving the system of coupled ordinary differential equations \eqref{eqn:Sl} \& \eqref{eqn:Su} along with the definition of flux in the saturated region \eqref{eq:harmonic-two-layer-final} gives the location of the shocks and the flux in the saturated region. Similar to the case of two-layered soils in \cite{shadab2022analysis}, to find the analytic value of $q_s(\tau)$, the ratio of shock speeds can be considered a constant as an ansatz given by
\begin{align}\label{eq:shock-speed-ratio}
    \frac{\Lambda_{{\mathscr{S}^*_1}}}{\Lambda_{{\mathscr{S}_3}}} = \frac{\zeta_U(\tau)}{\zeta_L(\tau)} = \mathfrak{R}\quad\mbox{for} \quad 0 < \tau \leq \tau_p.
\end{align}
where $\mathfrak{R}$ is the constant ratio of shock speeds which is negative and $\tau_p$ is the dimensionless time of ponding when the upward moving shock reaches to the surface. In the special case when this jump condition happens to exist at the surface, $\mathfrak{R}=0$, $\tau_p=0$ and $q_s=\Kh_{i_2}$ which is saturated hydraulic conductivity of the second intermediate state. Otherwise, by substituting equations from the shock speed ratio definition \eqref{eq:shock-speed-ratio}, shock speeds \eqref{eqn:Su} and \eqref{eqn:Sl}, and flux in the saturated region \eqref{eq:harmonic-two-layer-final}, $\mathfrak{R}$ comes out to be the solution of a quadratic equation

\begin{align}\label{eq:shock-speed-ratio-final}
    \mathfrak{R}(\textbf{u}_\leftstate,\textbf{u}_\rightstate) &= \frac{-b - \sqrt{b^2 - 4ac}}{2a} \quad \text{where}\\
     a &= \left(\frac{1-\mathcal{C}_\leftstate}{1-\mathcal{C}_\rightstate} \right)\left[ 1 - \left(\frac{1-\mathcal{C}_\rightstate+\mathcal{H}_\rightstate} {1-\mathcal{C}_\leftstate+\mathcal{H}_\leftstate}\right)^m \left(\frac{\mathcal{H}_\rightstate}{1-\mathcal{C}_\rightstate+\mathcal{H}_\rightstate} \right)^n \right], \nonumber \\
    b &=  - \left(\frac{1-\mathcal{C}_\leftstate}{1-\mathcal{C}_\rightstate} \right)\left[ 1 - \left(\frac{\mathcal{H}_\rightstate}{1-\mathcal{C}_\rightstate+\mathcal{H}_\rightstate} \right)^n \right]+ \left(\frac{\mathcal{H}_\leftstate}{1-\mathcal{C}_\leftstate+\mathcal{H}_\leftstate} \right)^n - 1 \quad \text{ and } \nonumber \\
    c &= 1- \left(\frac{1-\mathcal{C}_\leftstate+\mathcal{H}_\leftstate} {1-\mathcal{C}_\rightstate+\mathcal{H}_\rightstate}\right)^m \left(\frac{\mathcal{H}_\leftstate}{1-\mathcal{C}_\leftstate+\mathcal{H}_\leftstate} \right)^n  .\nonumber
\end{align}
Subsequently, substituting Equations \eqref{eq:sat-hyd-cond-sat-regions} and \eqref{eq:shock-speed-ratio} in \eqref{eq:harmonic-two-layer-final} gives the time-invariant dimensionless flux in the saturated region $q_s$ during $0 < \tau \leq \tau_p$ as
 
\begin{align}\label{eq:qs_final}
    q_s = \frac{ \mathfrak{R}-1}{\mathfrak{R}/(1-\mathcal{C}_\leftstate+\mathcal{H}_\leftstate)^m - 1/(1-\mathcal{C}_\rightstate+\mathcal{H}_\rightstate)^m}.
\end{align}
 
 The time and space invariant flux $q_s$ leads to the result using Equations \eqref{eqn:Su} and \eqref{eqn:Sl} that the shock speeds are constant for $0< \tau \leq \tau_p$. Therefore, the shock locations before ponding vary linearly with time. Lastly, the full solution for this case can be summarized as

\begin{align}\label{eq:final-sol-contact-jump-shock-temperate}
    \textbf{u} &= \begin{cases}\textbf{u}_l, \quad &  \zeta < \Lambda_{\mathscr{S}^*_1} \\  \textbf{u}_{i_1}=\begin{bmatrix}1\\1-\mathcal{C}_\leftstate+\mathcal{H}_\leftstate \end{bmatrix}, \quad &  \Lambda_{\mathscr{S}^*_1} < \zeta / \tau < 0\\ 
    \textbf{u}_{i_2}=\begin{bmatrix}1\\1-\mathcal{C}_\rightstate+\mathcal{H}_\rightstate \end{bmatrix}, \quad &  0 < \zeta / \tau < \Lambda_{\mathscr{S}_3} \\ 
    \textbf{u}_r, \quad & \zeta / \tau>\Lambda_{\mathscr{S}_3} \end{cases}.
\end{align}
Here the shock speeds $\Lambda_{\mathscr{S}^*_1}$ and $\Lambda_{\mathscr{S}_3}$ depend on both left and right states due to formation of the saturated region as the flux is governed by an elliptic PDE. {This is the case of a temperate firn with more porous and more wet layer lies on top of less porous and less wet layer (Figure \ref{fig:saturated-region-formation-temperate}\textit{d}). In this case, the lower layer is unable to accommodate the flux of water from the upper layer and therefore leads to a rising perched water table as well as a wetting front (Figures \ref{fig:saturated-region-formation-temperate}\textit{e},\textit{f}). Figure \ref{fig:saturated-region-formation-temperate} shows an example of the same where the top layer has 50\% porosity and 40\% liquid water content ($\phi_{w,l}$) and the bottom layer is 30\% porous and has only 10\% LWC. These values correspond to left and right states being $\textbf{u}_l=(0.9,0.4)^T$ and $\textbf{u}_r=(0.8,0.1)^T$ respectively (Figures \ref{fig:saturated-region-formation-temperate}\textit{a}-\textit{c}).} Consequently, the first and second intermediate states, $\textbf{u}_{i_1}=(1,0.5)^T$ and $\textbf{u}_{i_2}=(1,0.3)^T$ respectively, lie in the expanding saturated region (Figures~\ref{fig:saturated-region-formation-temperate}\emph{d}-\emph{f}). The speeds of both fronts $\mathscr{S}_1^*$ and $\mathscr{S}_3$ are constant before ponding occurs. 

\begin{figure}
    \centering
    \includegraphics[width=\linewidth]{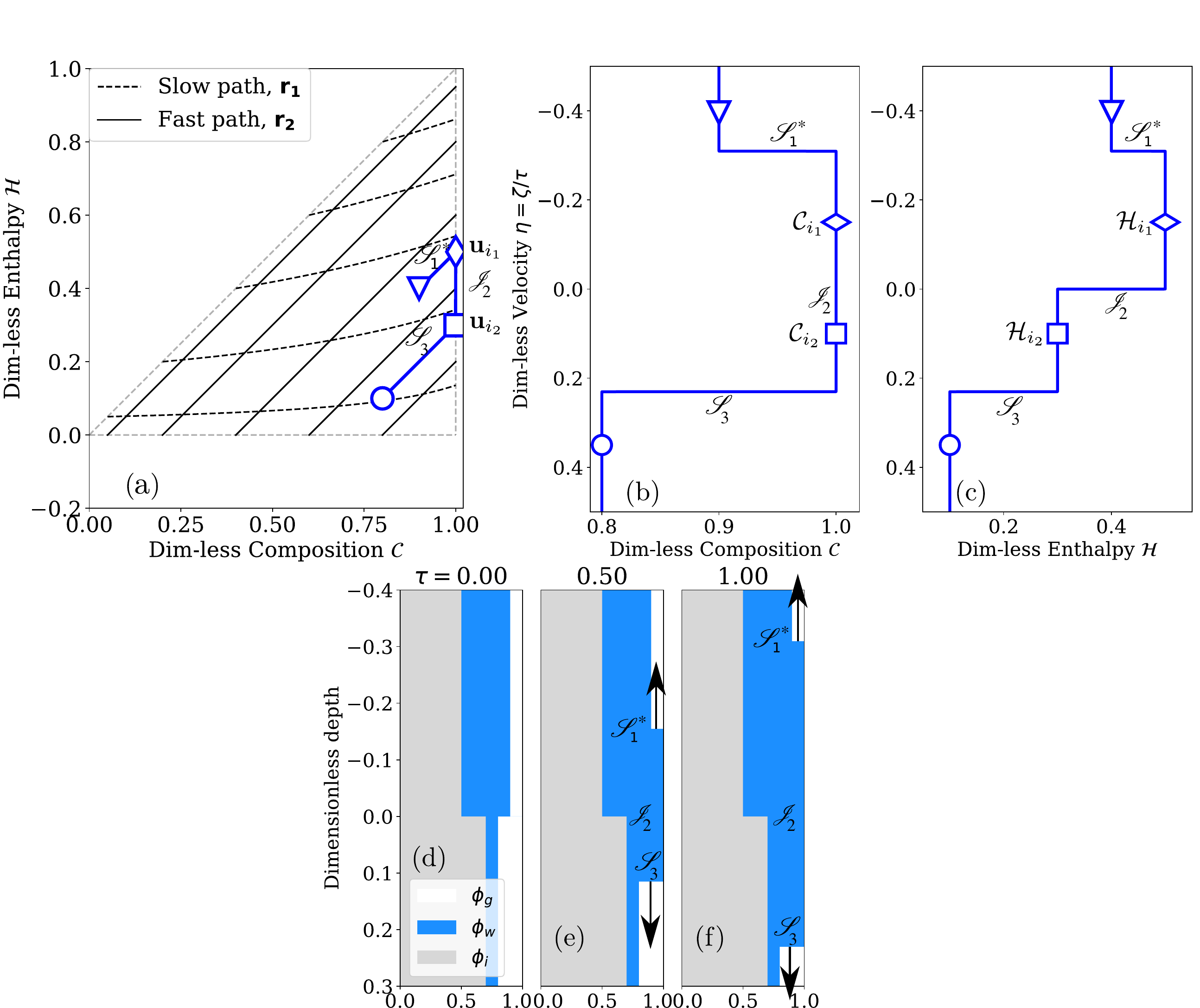}
    \caption{Formation of a fully-saturated region in temperate firn (Case VI): (a) Construction of solution in hodograph plane and their corresponding self-similar analytical solutions for (b) dimensionless composition and (c) dimensionless enthalpy with dimensionless velocity $\eta$. The result shown with dark blue line consists of a backfilling shock, $\mathscr{S}^*_1$, a jump, $\mathscr{J}_2$, and another wetting shock, $\mathscr{S}_3$, along with two intermediate states $\textbf{u}_{i_1}=(\mathcal{C}_{i_1},\mathcal{H}_{i_1})^T$ and $\textbf{u}_{i_2}=(\mathcal{C}_{i_2},\mathcal{H}_{i_2})^T$. The left and right states are $\textbf{u}_l=(0.9,0.4)^T$ and $\textbf{u}_r=(0.8,0.1)^T$ respectively. The first and second intermediate states, $\textbf{u}_{i_1}=(1,0.5)^T$ and $\textbf{u}_{i_2}=(1,0.3)^T$ respectively. The evolution of the volume fractions of the three phases in the resulting system at dimensionless times (d) $\tau=0$, (e) $0.5$, (f) $1.0$.}
    \label{fig:saturated-region-formation-temperate}
\end{figure}

\subsubsection{\textbf{Region 1 only (Ice and gas region), Case VII}}
In region 1 ($\mathcal{H}\leq 0$), since both fluxes are zero the system is not strictly hyperbolic and leads to a single wave. As the characteristic speed, $\lambda_p$, is constant, this wave is linearly degenerate and since $\lambda_1=\lambda_2=0$, the characteristic is stationary. The resulting wave is a stationary contact discontinuity $\mathscr{C}$. There won't be any transport between the two states since the fluxes of both composition and enthalpy are zero on either side. In other words, $\textbf{u}_\leftstate$ will be connected to $\textbf{u}_\rightstate$ by a stationary contact discontinuity $\mathscr{C}$ as

\begin{align*}
\textbf{u}_\leftstate\xrightarrow{\mathscr{C}}\textbf{u}_\rightstate.
\end{align*}
The resulting solution thus takes the form 

\begin{align}\label{eq:contact-disc-only-region1}
    \textbf{u} &= \begin{cases}\textbf{u}_l, \quad & \zeta < 0 \\  \textbf{u}_r, \quad & \zeta>0\end{cases}.
\end{align}

\begin{figure}
    \centering
    \includegraphics[width=\linewidth]{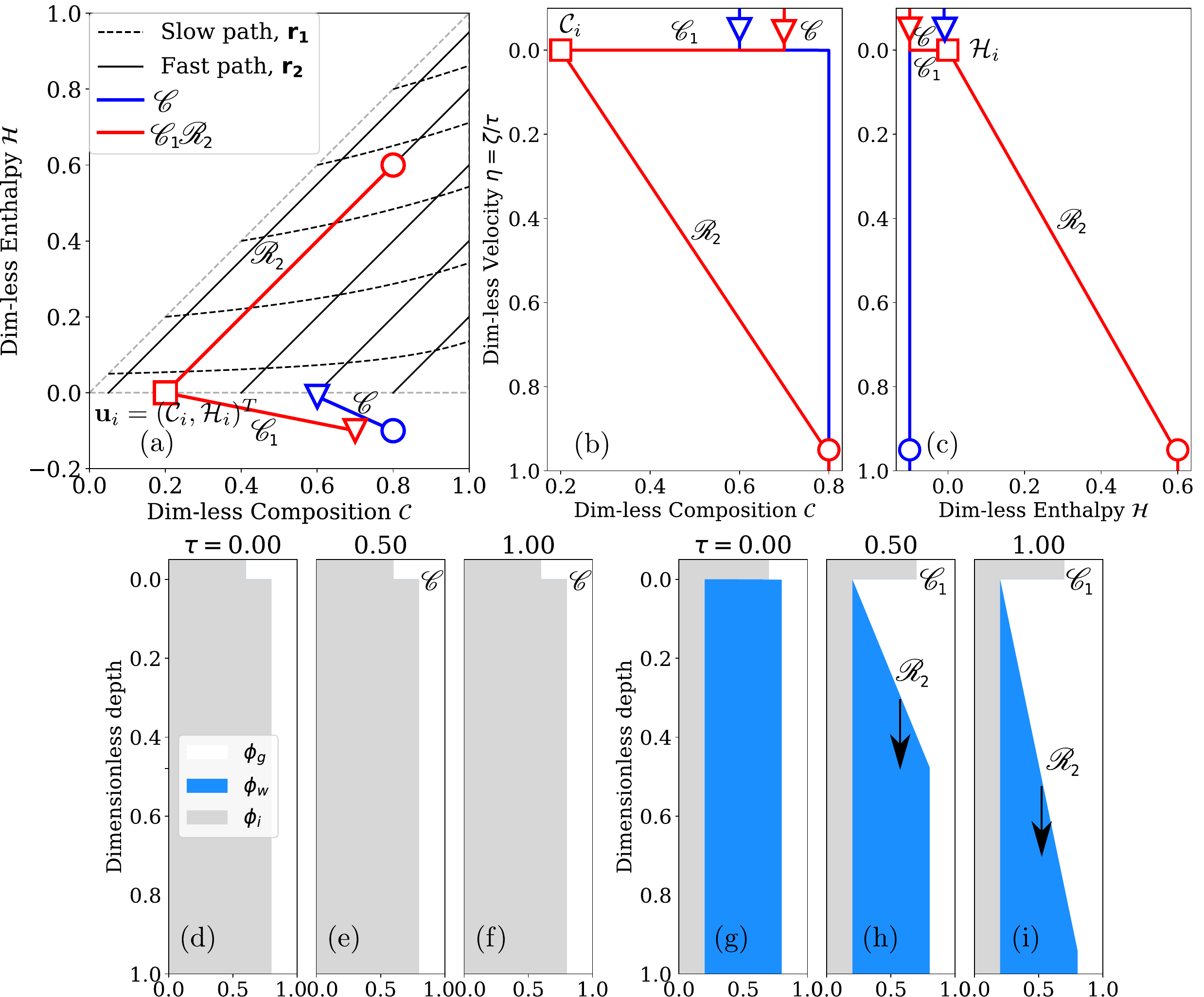}
    \caption{Solutions when the left state lies in Region 1 (ice and gas): either only a contact discontinuity appears (Case VII) or an intermediate state, $\textbf{u}_i$, along with a rarefaction wave $\mathscr{R}_2$ also forms in a $\mathscr{C}_1\mathscr{R}_2$ fashion (Case VIII). (a) Construction of solution in hodograph plane and their corresponding self-similar analytical solutions for (b) dimensionless composition and (c) dimensionless enthalpy with dimensionless velocity $\eta$. Blue and red lines respectively show the solutions when the right states are in Regions 1 and 2 respectively. The evolution of the volume fractions of the three phases in the resulting system at dimensionless times $\tau=0,0.5,1.0$ for the two configurations: (d-f) Case VII - $\mathscr{C}$ and (g-i) Case VIII - $\mathscr{C}_1\mathscr{R}_2$.}
    \label{fig:region-2-related}
\end{figure}

This corresponds to the case when a more porous, cold firn lies on top of a less, cold porous firn (or glacial ice) or vice versa which results in a dry, static system in this model. An example for a temperate, 40\% porous firn on top of a cold ($T=-19.8^\circ$C), 20\% porous firn corresponding to $\textbf{u}_l=(0.6,0.0)^T$ and $\textbf{u}_r=(0.8,-0.1)^T$, as shown in Figures~\ref{fig:region-2-related}\textit{a}-\textit{c} (blue lines) and \ref{fig:region-2-related}\textit{d}-\textit{f}. However, in reality the firn may compact due to overburden \citep{cuffey2010physics} which is not considered in the present model. Note that the temperatures (not shown) in the two layers are the same as the initial condition with sharp transition at $\zeta=0$ due to the absence of heat diffusion.

\subsection{\textbf{Inter-regional cases}}
\subsubsection{\textbf{From Region 1 (ice and gas) to Region 2 (three-phase region), Case VIII}} \label{reg1to2}
  In this case, the left state lies in region 1 corresponding to the cold firn and the right state resides in the three-phase region (region 2). The result is a contact discontinuity $\mathscr{C}_1$ onto the solidus ($\mathcal{H}=0$) where the intermediate state lies (see Figure \ref{fig:region-2-related}\textit{a} for example), i.e.,
  
\begin{align}\label{eq:R1toR2-Hi}
    \mathcal{H}_i = 0.
\end{align}

Moreover, the intermediate state, $\mathcal{C}_i$, lies on the fast path (constant porosity line) in region 2 satisfying Equation \eqref{eq:fast-path} which gives

\begin{align}\label{eq:R1toR2-Ci}
    \mathcal{C}_i &= \mathcal{C}_\rightstate - \mathcal{H}_\rightstate    .
\end{align}

Simultaneously, the intermediate state is connected to the right state on the fast path, resulting in a rarefaction wave $\mathscr{R}_2$. It is important to note that the second wave, $\mathscr{W}_2$, is supposed to be faster than the first and that is why the slow path is avoided in the three-phase region (region 2). The final solution to this case is the same as given in Equation \eqref{eq:final-sol-contact-rarefaction}.

{As an example, the left state corresponds to a 30\% porous cold layer at T$=-22.63^\circ$C lying on top of the layer corresponding to wet, temperate firn similar to liquid storage with 80\% porosity and 60\% LWC ($\phi_{w,r}$), as shown in Figure \ref{fig:region-2-related}\textit{g}. This configuration corresponds to $\textbf{u}_l=(0.7,-0.1)^T$ and $\textbf{u}_r=(0.8,0.6)^T$ with intermediate state $\textbf{u}_i=(0.2,0.0)^T$ (Figure \ref{fig:region-2-related}\textit{a}, red line). As a result, the porosity jump remains stationary but the liquid storage drains downwards due to gravity forming a self-similar rarefaction wave as shown in Figures~\ref{fig:region-2-related}\emph{b}-\emph{c} (red lines) and \ref{fig:region-2-related}\emph{g}-\emph{i}.} In a nutshell, this case describes the evolution of a more saturated firn layer below a previously formed less permeable, cold frozen fringe.


\subsubsection{\textbf{From Region 2 (three-phase region) to Region 1 (ice and gas)}}\label{reg2to1}

This case corresponds to porous and wetter firn at melting temperature with liquid meltwater above the cold firn initially (see Figures~\ref{fig:region-refreezing-front-unsat}-\ref{fig:ice-lens-formation} for example). In this case, the left state corresponding to the top layer lies in the three-phase region (region 2) and the right state corresponding to the bottom layer lies in region 1. Note that the temperature remains subzero only in the cold region (region 1, $\mathcal{H}<0$) which lies only in the right state with a sharp transition from the left or intermediate state as heat diffusion is not considered in this model. There can be four results corresponding to this initial condition: \\

\begin{figure}
    \centering
    \includegraphics[width=\linewidth]{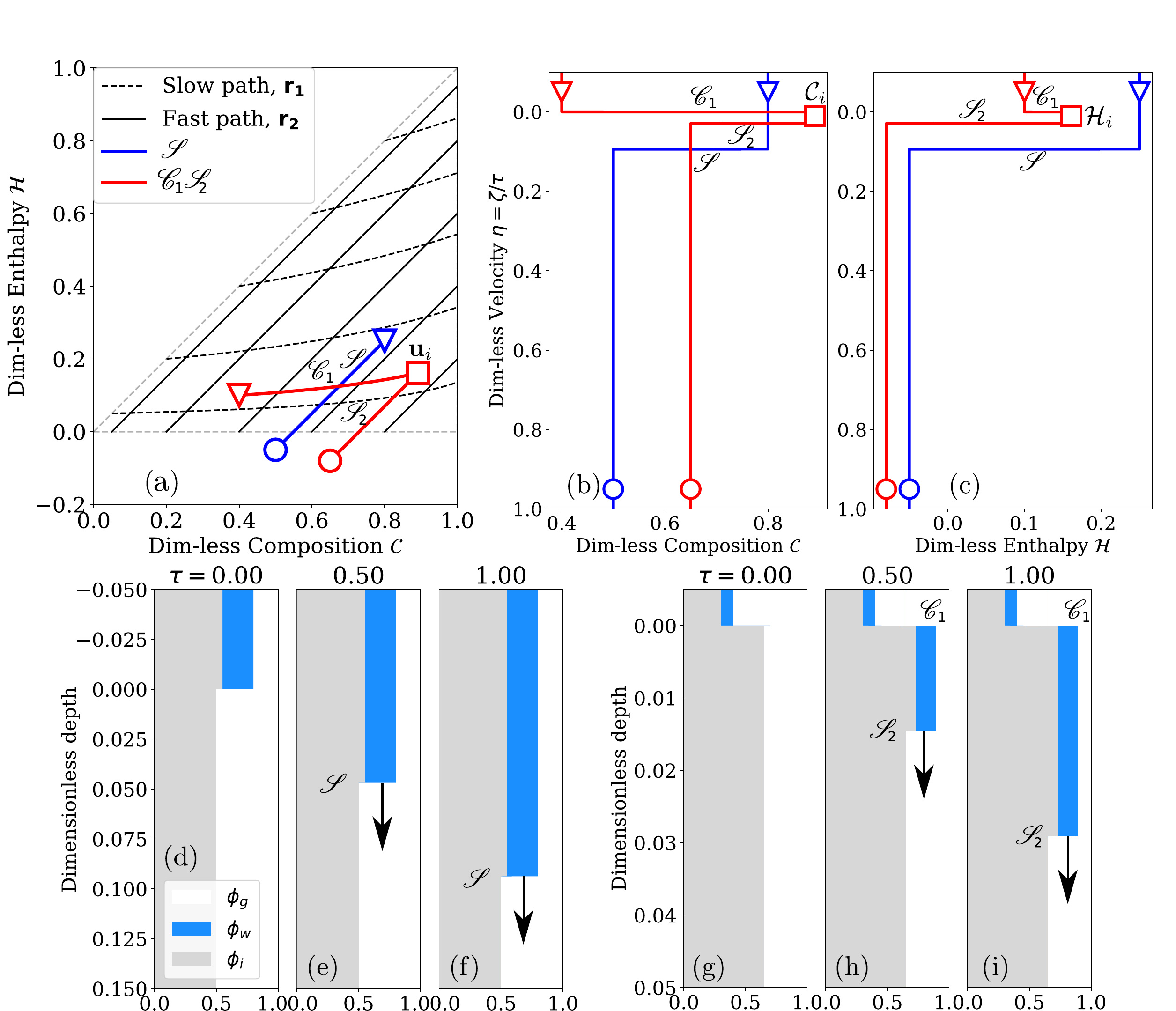}
    \caption{Solutions when the left state lies in region 2 (three-phase) and the right state lies in region 1 without formation of saturated regions: either only a refreezing front $\mathscr{S}$ (Case IX) appears or a contact discontinuity, $\mathscr{C}_1$, an intermediate state, $\textbf{u}_i$, along with a refreezing front $\mathscr{S}_2$ forms in a $\mathscr{C}_1\mathscr{S}_2$ (Case X) fashion. (a) Construction of solution in hodograph plane and their corresponding self-similar analytical solutions for (b) dimensionless composition and (c) dimensionless enthalpy with dimensionless velocity $\eta$. Red and blue lines respectively show the solutions when the right states are in region 2 and region 1 respectively. The evolution of the volume fractions of the three phases in the resulting system at dimensionless times $\tau=0,0.5,1.0$ for the two configurations: (d-f) Case VII - $\mathscr{C}$ and (g-i) Case VIII - $\mathscr{C}_1\mathscr{R}_2$.}
    \label{fig:region-refreezing-front-unsat}
\end{figure}

\subparagraph{(i.) Shock (Case IX):} When the right state in region 1 lies on the fast path in region 2 (three-phase region) extended to region 1 (ice and gas region) referred to as \textit{extended fast path}, it results in only a single moving shock $\mathscr{S}$ (see Figure~\ref{fig:region-refreezing-front-unsat}\textit{a}, blue line). Note that the extended fast paths are not the constant porosity contours in Region 1 (Figures~\ref{fig:combined-variables}\emph{c},\emph{g}). Mathematically, the left and right states are connected by the relation

    \begin{align}\label{eq:R2toR1shock-only}
    \mathcal{C}_\rightstate&=\mathcal{H}_\rightstate+ \mathcal{C}_\leftstate-\mathcal{H}_\leftstate .
\end{align}

The shock speed $\Lambda_{\mathscr{S}}$ can then be evaluated using Rankine-Hugoniot condition \eqref{eq:RH-condition-shock-case}. The final solution of this case is provided in Equation \eqref{eq:final-sol-shock-only}.
This case corresponds to the part where meltwater percolates and precipitates directly into ice, reducing the porosity equal to the left state thereby extending the frozen fringe (Figures~\ref{fig:region-refreezing-front-unsat}\emph{d}-\emph{f}). Physically, the water seeps from the top layer to the bottom layer due to gravity and a part of it precipitates due to heat loss to the surrounding cold firn. Since the flux in the right state is zero, the evaluation depends on the difference of either dimensionless composition or enthalpy. {Figures~\ref{fig:region-refreezing-front-unsat}\emph{a}-\emph{c} (blue line) shows an example of the solution where temperate firn with 45\% porosity ($\varphi_{l}$) and 25\% LWC ($\phi_{w,l}$) initially lies on a cold layer of porosity 50\% at T$=-15.84^\circ$C (Figure~\ref{fig:region-refreezing-front-unsat}\textit{d}). These initial conditions correspond to a left state, $\textbf{u}_l=(0.8,0.25)^T$ and a right state $\textbf{u}_r=(0.5,-0.05)^T$ (Figures~\ref{fig:region-refreezing-front-unsat}\textit{a}-\textit{c}, blue line). A refreezing shock $\mathscr{S}$ results that moves downwards while warming up the surrounding snow by partially refreezing (Figures~\ref{fig:region-refreezing-front-unsat}\emph{d}-\emph{f}). This reduces the porosity behind the wetting front from $50\%$ to $45\%$, equal to the same value as the left state thereby extending the previously formed frozen fringe in the top layer into the bottom layer.}

\subparagraph{(ii.) 1-Contact discontinuity and 2-Shock (Case X):}
This case occurs when the left state cannot be directly connected to the right state along the extended fast path as discussed in Case IX. So in this case, the left state $\textbf{u}_\leftstate$ in region 2 is more porous and connected to an intermediate state, $\textbf{u}_i$, along the slow path where the dimensionless composition of the intermediate state, $\mathcal{C}_i$, is less than unity, to keep the medium unsaturated (see Figure~\ref{fig:region-refreezing-front-unsat}\textit{a}, red line). Similar to the only shock $\mathscr{S}$ case (Case IX), the intermediate state lies on the extended fast path from region 2 to region 1. Therefore the solution is a combination of a stationary contact discontinuity $\mathscr{C}_1$ and a moving shock $\mathscr{S}_2$. The first characteristic wave, contact discontinuity $\mathscr{C}_1$, satisfies Equation \eqref{eq:slow-path} for left state $\textbf{u}_\leftstate$ and intermediate state $\textbf{u}_i$. Since the second characteristic lies on the Hugoniot locus, a shockwave $\mathscr{S}_2$ results, that satisfies the Hugoniot-jump condition \eqref{eq:RH-condition} for the intermediate state $\textbf{u}_i$ and right state $\textbf{u}_\rightstate$. Combining all these equations gives a simple relation for dimensionless enthalpy and composition at the intermediate state as 

\begin{align}
\mathcal{H}_i=\mathcal{H}_\leftstate\left(\frac{1-\mathcal{C}_\leftstate+\mathcal{H}_\leftstate}{1-\mathcal{C}_\rightstate+\mathcal{H}_\rightstate}\right)^{\frac{m-n}{n}} \quad \text{and} \quad \mathcal{C}_i=\mathcal{H}_i+ \mathcal{C}_\rightstate-\mathcal{H}_\rightstate . \label{eq:60new}
\end{align}

The dimensionless velocity of the shock wave $\mathscr{S}_2$ (refreezing front) is

\begin{align}
    \Lambda_{\mathscr{S}_2} &= \frac{f (\textbf{u}_i) - f(\textbf{u}_\rightstate)}{\mathcal{H}_i-\mathcal{H}_\rightstate} = \frac{\mathcal{H}_i^n(1-\mathcal{C}_i+\mathcal{H}_i)^{m-n}}{\mathcal{H}_i-\mathcal{H}_\rightstate} \quad \textrm{(since $f(\textbf{u}_\rightstate)=0$)}. \label{eq:62}
\end{align}

{The full solution is given in Equation \eqref{eq:final-sol-contact-shock} with shock speed \eqref{eq:62}. Since $\mathcal{H}_r<0$, $\mathcal{H}_i>0$ and $f(\textbf{u}_r)=0$, the speed of the refreezing front $\mathscr{S}_2$ is slower than temperate firn case (Case V) due to refreezing. Figures~\ref{fig:region-refreezing-front-unsat}\emph{a}-\emph{c} (red line) shows an example of such a solution for light rainfall on a multilayered firn with coarse to fine transition. The left state corresponds to the wetter, temperate firn on top with a porosity of $70\%$ and a liquid water content $\phi_{w,l}=0.1$. The right state corresponds to a cold and dry firn layer with 35\% porosity and a temperature of T$=-19.49^\circ$C. These values for left and right states correspond to states being $\textbf{u}_l=(0.4,0.1)^T$ and $\textbf{u}_r=(0.65,-0.08)^T$, respectively. Here, a stationary contact discontinuity $\mathscr{C}_1$ stays at the surface leading to a growing intermediate state $\textbf{u}_i=(0.893,0.163)$ formed due to partial refreezing of the rainwater which warms the surrounding firn (Figures~\ref{fig:region-refreezing-front-unsat}\emph{g}-\emph{i}). The intermediate state with $27\%$ porosity and $16.3\%$ LWC expands with time as the refreezing front $\mathscr{S}_2$ infiltrates further into the right state. Note that the porosity in the intermediate state is smaller than the right state which decreases further due to meltwater refreezing leading to the formation of a fresh frozen fringe. This phenomenon has been observed in the model by \cite{meyer2017continuum} who used the dimensional form of Equation \eqref{eq:62} with distinct densities of ice and water phases to explain the field data from \cite{humphrey2012thermal}.}

\subparagraph{(iii.) 1-Backfilling shock, 2-Jump and 3-Shock (formation of a saturated region, Case XI):}
\begin{figure}
    \centering
    \includegraphics[width=\linewidth]{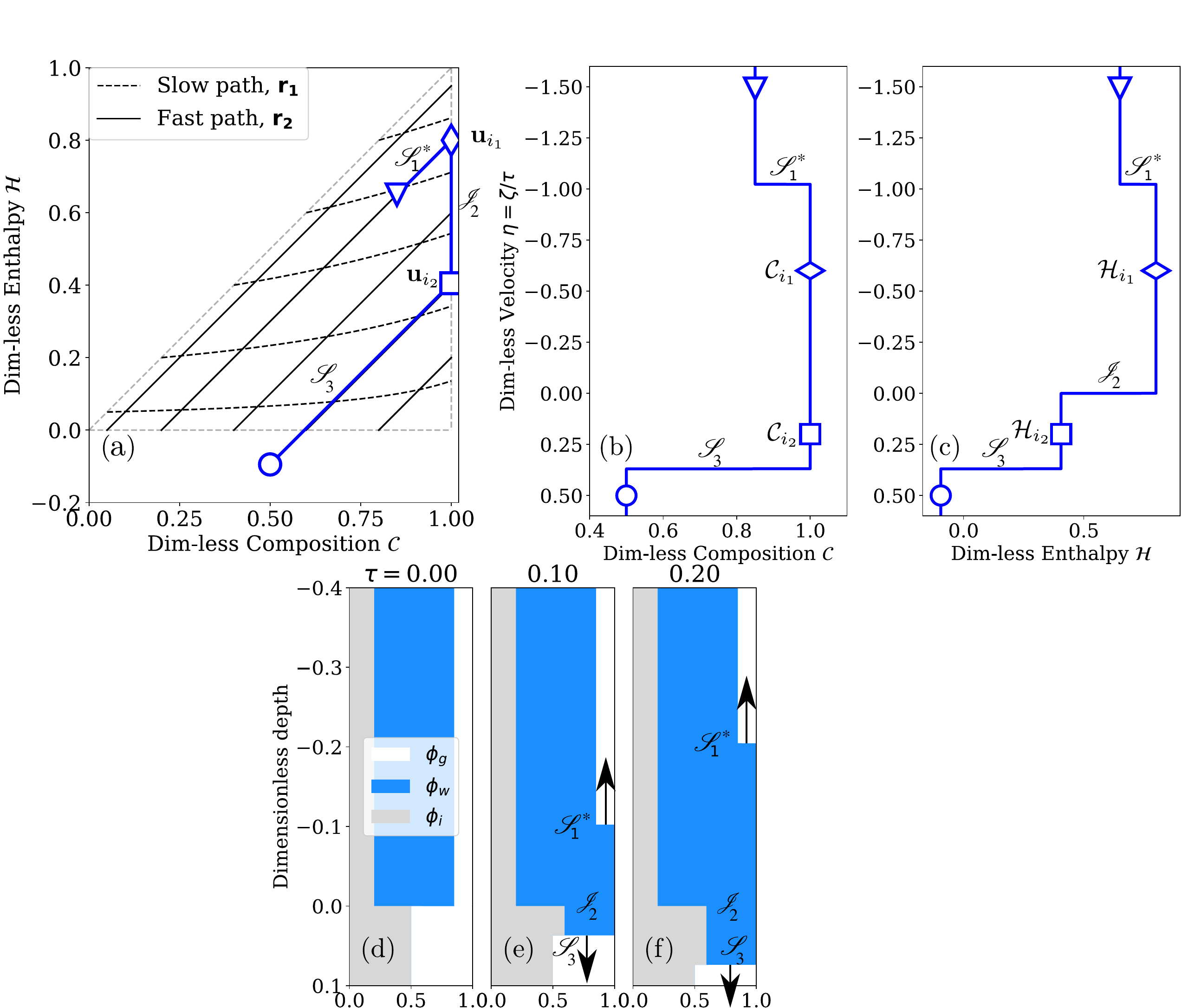}
    \caption{Formation of a fully-saturated region when right state lies in region 1: (a) Construction of solution in hodograph plane and their corresponding self-similar analytical solutions for (b) dimensionless composition and (c) dimensionless enthalpy with dimensionless velocity $\eta$. The result shown with dark blue line consists of a backfilling shock, $\mathscr{S}^*_1$, a jump, $\mathscr{J}_2$, and another refreezing shock, $\mathscr{S}_3$, along with two intermediate states $\textbf{u}_{i_1}=(\mathcal{C}_{i_1},\mathcal{H}_{i_1})^T$ and $\textbf{u}_{i_2}=(\mathcal{C}_{i_2},\mathcal{H}_{i_2})^T$. The evolution of the volume fractions of the three phases in the resulting system at dimensionless times (d) $\tau=0$, (e) $0.1$, (f) $0.2$.}
    \label{fig:saturated-region-formation-cold}
\end{figure}

Here the slow path originating from the left state and the extended fast path emanating from the right state do not intersect where $\mathcal{C}<1$, making it a different case than Case X. If the intermediate state lies on the saturated region, then the hyperbolic nature of the solution breaks down, similar to Case VI for temperate firn (see Table \ref{table:2}). So the state solution \eqref{eq:final-sol-contact-jump-shock-temperate} is the same as provided for the temperate firn (Case VI) consisting of a backfilling shock $\mathscr{S}^*_1$ moving upwards with speed \eqref{eqn:Su}, a stationary jump $\mathscr{J}_2$ at the location of initial jump $\zeta=0$ and a ``refreezing'' front $\mathscr{S}_3$ moving downwards with speed provided in Equation \eqref{eqn:Sl}. The flux in the saturated region $q_s$ is again provided by 
 Equation \eqref{eq:harmonic-two-layer-final}. Note that the flux at the right state $f(\textbf{u}_r)=0$ since it lies in region 1. Invoking the ansatz for a constant shock speed ratio, similar to Case VI, provides an analytic relation for a constant flux in the saturated region, $q_s$, same as Equation \eqref{eq:qs_final} with both shocks moving in opposite directions at constant speeds. However since the flux at the right state is zero, the relation for the ratio of shock speeds, $\mathfrak{R}$, earlier provided by Equation \eqref{eq:shock-speed-ratio-final}, now simplifies to 
\begin{align}\label{eq:shock-speed-ratio-final-refreezing}
    \mathfrak{R}(\textbf{u}_\leftstate,\textbf{u}_\rightstate) &= \frac{-b - \sqrt{b^2 - 4ac}}{2a} \quad \text{where}\\
     a &= \left(\frac{1-\mathcal{C}_\leftstate}{1-\mathcal{C}_\rightstate} \right), \nonumber \\
    b &=  - \left(\frac{1-\mathcal{C}_\leftstate}{1-\mathcal{C}_\rightstate} \right)+ \left(\frac{\mathcal{H}_\leftstate}{1-\mathcal{C}_\leftstate+\mathcal{H}_\leftstate} \right)^n - 1 \quad \text{ and } \nonumber \\
    c &= 1- \left(\frac{1-\mathcal{C}_\leftstate+\mathcal{H}_\leftstate} {1-\mathcal{C}_\rightstate+\mathcal{H}_\rightstate}\right)^m \left(\frac{\mathcal{H}_\leftstate}{1-\mathcal{C}_\leftstate+\mathcal{H}_\leftstate} \right)^n  .\nonumber
\end{align}

Figure~\ref{fig:saturated-region-formation-cold} shows an example of this case where the left state corresponds to wet and temperate layer with 80\% porosity ($\varphi_l$) and 65\% LWC ($\phi_{w,l}$) lying on top of very cold and less porous (T$=-30^\circ$C, $\varphi_r=50\%$) firn. This corresponds to the left and right states being $\textbf{u}_l=(0.85,0.65)^T$ and $\textbf{u}_r=(0.5,-0.095)^T$. The first and second intermediate states are $\textbf{u}_{i_1}=(1,0.8)^T$ and $\textbf{u}_{i_2}=(1,0.405)^T$ respectively. The upper shock, also called rising perched water table, is almost four times faster than the lower, refreezing front (see Figures~\ref{fig:saturated-region-formation-cold}\emph{b}-\emph{f}). Below the initial jump at $\zeta=0$ the porosity is initially reduced which is further reduced by refreezing, leading to formation of a frozen fringe (Figures~\ref{fig:saturated-region-formation-cold}\emph{e}-\emph{f}). The second intermediate state below jump $\mathscr{J}_2$ with a newly formed frozen fringe is unable to accommodate the whole volumetric flux of water, leading to the formation of a rising perched water table (see Figures~\ref{fig:saturated-region-formation-cold}\emph{e}-\emph{f}). Once the rising perched water table reaches the surface, it will lead to ponding and can eventually form runoff. 

\subparagraph{(iv.) 1-Backfilling shock, 2-Jump, 3-Contact discontinuity (formation of impermeable ice lens, Case XII):} This final case captures the formation of impermeable ice lens through convection of meltwater and associated latent heat. Ice lens forms through Case XI, if the extended fast path ($\mathcal{H}=\mathcal{C}+\mathfrak{C}$) emanating from the right state reach $\mathcal{C}\geq 1$ at the solidus ($\mathcal{H}=0$) or the porosity of the second intermediate state reaches 0. This case is entirely governed by the right state and ice lens will form if and only if the right state satisfies 
\begin{equation}\label{eq:ice-lens-condition}
    1-\mathcal{C}_r+\mathcal{H}_r\leq 0.
\end{equation} 
The exact location of the right state does not matter if it lies in the region of ice lens formation (grey region in Figure~\ref{fig:ice-lens-formation}\emph{a}). In other words, an impermeable ice lens via heat convection will form if and only if the cold content of the firn exceeds the latent enthalpy of incoming meltwater. The mathematical condition \eqref{eq:ice-lens-condition} in dimensional form has been given in \cite{humphrey2021physical} for distinct densities of ice and water. Furthermore, this condition has been qualitatively discussed in \cite{colbeck1976analysis,clark2017analytical}. In this limit, the solution for Case XI breaks down due to the formation of an impermeable ice lens. {Figure~\ref{fig:ice-lens-formation} shows the region of ice lens formation which is perpetuated by very low firn porosity or temperatures in the right state. The final solution for this case is a backfilling shock, $\mathscr{S}^*_1$ to the intermediate saturated region $i_1$. The first saturated region is connected to the second intermediate state $i_2$ through the jump $\mathscr{J}_2$. The second intermediate state $i_2$ is the infinitesimally thin ice lens (see Figures~\ref{fig:ice-lens-formation}\emph{e}-\emph{f}) which blocks the further meltwater. The solution in this case can be written as} 

\begin{align*}
\textbf{u}_\leftstate\xrightarrow{\mathscr{S}^*_1}\textbf{u}_{i_1}\xdashrightarrow{\mathscr{J}_2}\textbf{u}_{i_2} \xrightarrow{\mathscr{C}_3} {\textbf{u}_\rightstate}
\end{align*}
where the state solution is quantified as

\begin{align}\label{eq:ice-lens-formation-solution}
    \textbf{u} &= \begin{cases}\textbf{u}_l, \quad &  \zeta < \Lambda_{\mathscr{S}^*_1} \\  \textbf{u}_{i_1}=\begin{bmatrix}1\\1-\mathcal{C}_\leftstate+\mathcal{H}_\leftstate \end{bmatrix}, \quad &  \Lambda_{\mathscr{S}^*_1} < \zeta / \tau < 0\\ 
    \textbf{u}_{i_2}=\begin{bmatrix}1\\0 \end{bmatrix}, \quad &  0 < \zeta  < \d \zeta \\ 
    \textbf{u}_r, \quad & \zeta>\d \zeta \end{cases},
\end{align}
where $\d \zeta$ is the infinitesimally small thickness of the ice lens.

In this case, only a single shock moves upwards similar to the filling of a bucket. The flux in the saturated region is simply zero from the mass balance at the ice lens, i.e., $q_s=f(\textbf{u}_{i_1})=f(\textbf{u}_{i_2})=0$. The speed of the rising perched water table (first shock $\mathscr{S}^*_1$) is then
\begin{align}\label{eq:ice-lens-shock-speed-backfill}
    \Lambda_{\mathscr{S}^*_1}(\textbf{u}_l) = \frac{\d \zeta_U}{\d \tau} = \frac{f(\textbf{u}_l)}{\mathcal{C}_l-1} = \frac{\mathcal{H}^n(1-\mathcal{C}_l+\mathcal{H}_l)^{(m-n)}}{\mathcal{C}_l-1}
\end{align}
which depends only on the left-state variables.

\begin{figure}
    \centering
    \includegraphics[width=\linewidth]{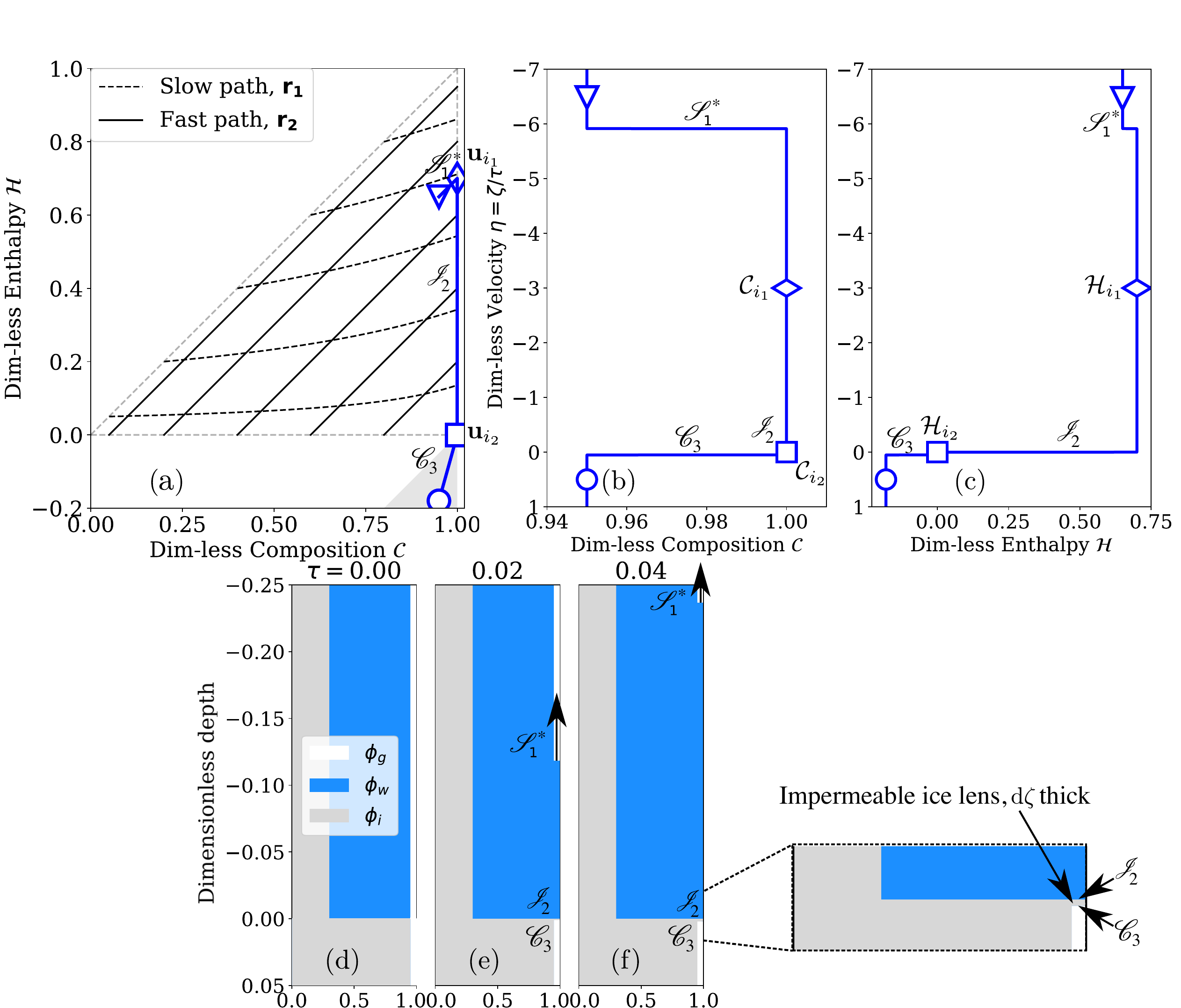}
    \caption{Impermeable ice lens formation: (a) Construction of solution in hodograph plane and their corresponding self-similar analytical solutions for (b) dimensionless composition and (c) dimensionless enthalpy with dimensionless velocity $\eta$. The result shown with dark blue line consists of a backfilling shock, $\mathscr{S}^*_1$, a jump, $\mathscr{J}_2$ and a contact discontinuity, $\mathscr{C}_3$ 
along with two intermediate states $\textbf{u}_{i_1}=(\mathcal{C}_{i_1},\mathcal{H}_{i_1})^T$ and $\textbf{u}_{i_2}=(\mathcal{C}_{i_2},\mathcal{H}_{i_2})^T$.  The second intermediate state $\textbf{u}_{i_2}$ corresponds to the impermeable ice lens. The grey region corresponds to $1-\mathcal{C}+\mathcal{H}\leq 0$ where the right state resides to cause impermeable ice lens formation. The evolution of the volume fractions of the three phases in the resulting system at dimensionless times (d) $\tau=0$, (e) $0.02$, (f) $0.04$.}
    \label{fig:ice-lens-formation}
\end{figure}
It can be deduced that either a cold firn (less $\mathcal{H}$) or a nearly non-porous firn with $\mathcal{C}$ close to unity will induce the formation of an impermeable ice lens via refreezing caused by heat advection. For example, Figure \ref{fig:ice-lens-formation} shows such a solution when wet and temperate layer of 70\% porosity and 65\% LWC on top of a cold layer with 5\% porosity at T$=-30^\circ$C. These conditions correspond to $\textbf{u}_l = (0.95,0.65)^T$ and $\textbf{u}_r = (0.95,-0.180)^T$ (Figures~\ref{fig:ice-lens-formation}\textit{a}-\textit{c}). Here the intermediate states lie at $\textbf{u}_{i_1}=(1,0.7)^T$ and $\textbf{u}_{i_2}=(1,0)^T$ respectively. An impermeable ice lens forms as a result (see Figures \ref{fig:ice-lens-formation}\emph{e}-\emph{f}) that blocks the flow of meltwater downwards. Since there is no meltwater percolating in the lower layer, the backfilling occurs very rapidly.

\begin{sidewaystable}
\caption{Summary of all the analytical solutions presented in this paper along with related works which have either studied or observed the corresponding scenario.}\label{table:2}
\begin{tabular}{|c|cc|p{4.5cm}|p{1.0cm}p{1.0cm}p{1.5cm}| p{5.5cm}| p{3cm}| }
\hline
\multirow{2}{*}{Case} & \multicolumn{2}{c|}{Region} & \multirow{2}{*}{\hspace{1.2cm}State properties} & \multicolumn{3}{c|}{Solution}  & \multirow{2}{*}{Physical relevance} & \multirow{2}{*}{Related work(s)} \\ \cline{2-3} \cline{5-7} 
    & \multicolumn{1}{c|}{$\textbf{u}_l$} & $\textbf{u}_r$ &        & \multicolumn{1}{c|}{Summary} & \multicolumn{1}{p{1.5cm}|}{Equation(s)} & Construction figure & & \\ \hline
I  & \multicolumn{1}{c|}{2}  &  2 &    $f(\textbf{u}_l)=f(\textbf{u}_r)$$\Leftrightarrow$$\textbf{u}_l $, $\textbf{u}_r$ connected by slow path    & \multicolumn{1}{c|}{$\textbf{u}_\leftstate\xrightarrow{\mathscr{C}}\textbf{u}_\rightstate $}  & \multicolumn{1}{p{1.5cm}|}{\ref{eq:slow-path}, \ref{eq:stationary-linear-rxn-front}, \ref{eq:contact-disc-only}} &  \ref{fig:contact_shock_rarefaction} (green line) & Steady meltwater flux inside a temperate firn with a jump in porosity. & $-$\\ \hline
II  & \multicolumn{1}{c|}{2}  &  2 &    $f(\textbf{u}_l)<f(\textbf{u}_r)$; $\textbf{u}_l $, $\textbf{u}_r$ connected by fast path; $\mathcal{C}_l < \mathcal{C}_r $    & \multicolumn{1}{c|}{$\textbf{u}_\leftstate\xrightarrow{\mathscr{R}}\textbf{u}_\rightstate $}  & \multicolumn{1}{p{1.5cm}|}{\ref{eq:fast-path}, \ref{eq:final-sol-rarefaction-only}} &  \ref{fig:contact_shock_rarefaction} (blue line) & A sudden decrease in meltwater flux inside a temperate firn with constant porosity. & \cite{colbeck1976analysis,singh1997kinematic,clark2017analytical} \\ \hline
III  & \multicolumn{1}{c|}{2}  &  2 &    $f(\textbf{u}_l)>f(\textbf{u}_r)$; $\textbf{u}_l $, $\textbf{u}_r$ connected by fast path; $\mathcal{C}_l > \mathcal{C}_r $    & \multicolumn{1}{c|}{$\textbf{u}_\leftstate\xrightarrow{\mathscr{S}}\textbf{u}_\rightstate $}  & \multicolumn{1}{p{1.5cm}|}{\ref{eq:fast-path}, \ref{eq:RH-condition-shock-case}, \ref{eq:final-sol-shock-only}} &  \ref{fig:contact_shock_rarefaction} (red line) & A sudden increase in meltwater flux inside a temperate firn with constant porosity. & \cite{colbeck1971one,colbeck1972theory,singh1997kinematic,samimi2021time} \\ \hline
IV  & \multicolumn{1}{c|}{2}  &  2 &    $f(\textbf{u}_l)<f(\textbf{u}_r)$; $\textbf{u}_l $, $\textbf{u}_r$ cannot be connected by either paths; $\mathcal{C}_i<1$   & \multicolumn{1}{c|}{$\textbf{u}_\leftstate\xrightarrow{\mathscr{C}_1}\textbf{u}_i\xrightarrow{\mathscr{R}_2}\textbf{u}_\rightstate $}  & \multicolumn{1}{p{1.5cm}|}{\ref{eq:52}, \ref{eq:final-sol-contact-rarefaction}} &  \ref{fig:two-intermediate-states} (blue line) &  A sudden decrease in meltwater flux inside a temperate firn with a step change in porosity. & $-$ \\ \hline
V  & \multicolumn{1}{c|}{2}  &  2 &    $f(\textbf{u}_l)>f(\textbf{u}_r)$; $\textbf{u}_l $, $\textbf{u}_r$ cannot be connected by either paths; $\mathcal{C}_i<1$   & \multicolumn{1}{c|}{$\textbf{u}_\leftstate\xrightarrow{\mathscr{C}_1}\textbf{u}_i\xrightarrow{\mathscr{S}_2}\textbf{u}_\rightstate $}  & \multicolumn{1}{p{1.5cm}|}{\ref{eq:intermediate-Hi-temperate-C1R2}, \ref{eq:final-sol-contact-shock}, \ref{eq:RH-condition-shock-case} } &  \ref{fig:two-intermediate-states} (red line) & A sudden increase in meltwater flux inside a temperate firn with a step change in porosity. & $-$ \\ \hline
VI  & \multicolumn{1}{c|}{2}  &  2 &    $f(\textbf{u}_l)>f(\textbf{u}_r)$; $\textbf{u}_l $, $\textbf{u}_r$ cannot be connected by the combination of two paths in Region 2; slow path emanating from $\textbf{u}_l $ does not intersect with fast path emanating from $\textbf{u}_r$ where $\mathcal{C}<1$  & \multicolumn{1}{c|}{$\textbf{u}_\leftstate\xrightarrow{\mathscr{S}^*_1}\textbf{u}_{i_1}\xdashrightarrow{\mathscr{J}_2}\textbf{u}_{i_2} \xrightarrow{\mathscr{S}_3}\textbf{u}_\rightstate$}  & \multicolumn{1}{p{1.5cm}|}{ \ref{eqn:Su}, \ref{eqn:Sl}, \ref{eq:harmonic-two-layer-final}, \ref{eq:sat-hyd-cond-sat-regions}, \ref{eq:shock-speed-ratio-final}, \ref{eq:qs_final}, \ref{eq:final-sol-contact-jump-shock-temperate}} &  \ref{fig:saturated-region-formation-temperate} & A sudden increase in meltwater flux inside a temperate firn with a step decline in porosity with depth leading to formation of a rising perched water table. & \cite{shadab2022analysis} \\ \hline
VII   & \multicolumn{1}{c|}{1}  &  1  &    $f(\textbf{u}_l)=f(\textbf{u}_r)=0$    & \multicolumn{1}{c|}{$\textbf{u}_\leftstate\xrightarrow{\mathscr{C}}\textbf{u}_\rightstate $}  & \multicolumn{1}{p{1.5cm}|}{\ref{eq:contact-disc-only-region1}} &  \ref{fig:region-2-related} (blue line) & A cold firn with a step change in porosity with no meltwater flux. & $-$ \\ \hline
VIII  & \multicolumn{1}{c|}{1}  &  2 &  $-$   & \multicolumn{1}{c|}{$\textbf{u}_\leftstate\xrightarrow{\mathscr{C}_1}\textbf{u}_i\xrightarrow{\mathscr{R}_2}\textbf{u}_\rightstate $}  & \multicolumn{1}{p{1.5cm}|}{ \ref{eq:final-sol-contact-rarefaction}, \ref{eq:R1toR2-Hi}, \ref{eq:R1toR2-Ci}} & \ref{fig:region-2-related} (red line) & A cold firn with no meltwater flux overlying a wet, temperate firn. & $-$\\ \hline
IX  & \multicolumn{1}{c|}{2}  &  1 &    $\textbf{u}_l $, $\textbf{u}_r$ connected by extended fast path & \multicolumn{1}{c|}{$\textbf{u}_\leftstate\xrightarrow{\mathscr{S}}\textbf{u}_\rightstate $}  & \multicolumn{1}{p{1.5cm}|}{\ref{eq:RH-condition-shock-case}, \ref{eq:final-sol-shock-only}, \ref{eq:R2toR1shock-only}} &  \ref{fig:region-refreezing-front-unsat} (blue line) & A sudden increase in meltwater flux into a cold firn with constant porosity leading to the formation of frozen fringe.  & \cite{colbeck1976analysis,clark2017analytical,meyer2017continuum} \\ \hline
X  & \multicolumn{1}{c|}{2}  &  1 &    $\textbf{u}_l $, $\textbf{u}_r$ cannot be connected by extended fast path, $\mathcal{C}_i<1$ & \multicolumn{1}{c|}{$\textbf{u}_\leftstate\xrightarrow{\mathscr{C}_1}\textbf{u}_i\xrightarrow{\mathscr{S}_2}\textbf{u}_\rightstate $}  & \multicolumn{1}{p{1.5cm}|}{\ref{eq:final-sol-contact-shock}, \ref{eq:60new}, \ref{eq:62}} &  \ref{fig:region-refreezing-front-unsat} (red line) & A sudden increase in meltwater flux into a cold firn with a step change porosity leading to the formation of frozen fringe. & \cite{colbeck1976analysis,clark2017analytical,meyer2017continuum} \\ \hline
XI  & \multicolumn{1}{c|}{2}  &  1 &    $\textbf{u}_l $, $\textbf{u}_r$ cannot be connected by extended fast path; slow path from $\textbf{u}_l $ does not intersect with extended fast path from $\textbf{u}_r$ where $\mathcal{C}<1$  & \multicolumn{1}{c|}{$\textbf{u}_\leftstate\xrightarrow{\mathscr{S}^*_1}\textbf{u}_{i_1}\xdashrightarrow{\mathscr{J}_2}\textbf{u}_{i_2} \xrightarrow{\mathscr{S}_3}\textbf{u}_\rightstate$}  & \multicolumn{1}{p{1.5cm}|}{ \ref{eqn:Su}, \ref{eqn:Sl}, \ref{eq:harmonic-two-layer-final}, \ref{eq:sat-hyd-cond-sat-regions}, \ref{eq:shock-speed-ratio}, \ref{eq:shock-speed-ratio-final-refreezing}, \ref{eq:qs_final}, \ref{eq:final-sol-contact-jump-shock-temperate}} &  \ref{fig:saturated-region-formation-cold} & A sudden increase in meltwater flux into a cold firn with a step decrease porosity leading to the formation of frozen fringe and a rising perched water table. & \cite{humphrey2021physical} \\ \hline
XII  & \multicolumn{1}{c|}{2}  &  1 &    $\textbf{u}_r$ lies in impermeable ice lens region satisfying $1-\mathcal{C}_r-\mathcal{H}_r\leq 0 $ & \multicolumn{1}{c|}{$\textbf{u}_\leftstate\xrightarrow{\mathscr{S}^*_1}\textbf{u}_{i_1}\xdashrightarrow{\mathscr{J}_2}\textbf{u}_{i_2} \xrightarrow{\mathscr{C}_3} {\textbf{u}_\rightstate}$}  & \multicolumn{1}{p{1.5cm}|}{ \ref{eq:ice-lens-formation-solution}, \ref{eq:ice-lens-shock-speed-backfill}} &  \ref{fig:ice-lens-formation} & A sudden increase in meltwater flux into a cold firn with a step decrease porosity leading to the formation an impermeable ice lens and a rising perched water table. & \cite{humphrey2021physical} \\ \hline
\end{tabular}
\end{sidewaystable}

\section{Infiltration into a multilayered firn - Validation with numerical simulations}\label{sec:validation}
Here we demonstrate the application of this theory to a realistic multilayered firn leading to the formation of a perched firn aquifer. This final test shows the infiltration into multilayered firn after a melt event combining two cases proposed in Section \ref{sec:melt-transport} that are summarized in Table \ref{table:2} (see Figure~\ref{fig:discussion}\emph{a}). This problem summates the commonly studied wetting front propagation in a temperate region (Case III) with the wetting front in a cold region and the formation of a perched water table (Case XI) that is not studied well in the literature. This problem demonstrates a delay in meltwater ponding at the surface due to a decay in both firn porosity and temperature with depth. The firn is 70\% porous, dry and at 0$^\circ$C, above a dimensionless depth of $\zeta=1$ (Figures~\ref{fig:discussion}\emph{a}). At time $\tau=0$, the meltwater is generated at the surface ($\zeta=0$), which keeps the liquid water content (denoted by $LWC$ or $\phi_w$) to 40\% making the firn very wet at the surface, $\zeta=0$. The analytic solution of this problem can be constructed in the hodograph plane in two parts (Figure~\ref{fig:discussion}\emph{b}). 

First, a wetting front $\mathscr{S}$ initially propagates downwards with a constant dimensionless speed $\Lambda_\mathscr{S}=0.28$, as given in Case III (Figures~\ref{fig:discussion}\emph{c}-\emph{e}, red dashed line). The dimensionless time when the initial front reaches the depth of transition $\zeta_t = 1$ to form a saturated region is $\tau_s = \frac{\zeta_t}{\Lambda_\mathscr{S}}=1/0.28=3.57$, as shown in Case XI. Afterwards the saturated region forms due to large meltwater flux compared to the hydraulic conductivity of the second intermediate region formed due to refreezing below the jump $\zeta>1$, making the medium pressure driven instead and enforcing $q_s(\tau) \leq K_{i_2}$. This saturated region expands in both directions as a perched water table (upper shock shown by green dashed line in Figure~\ref{fig:discussion}) rising to the surface and the wetting front ${\mathscr{S}_3}$ that percolates in the cold region (lower shock shown with blue dashed line) while refreezing a part of meltwater to warm the surrounding snow. Thus, it shows a reduction in porosity from 30\% to 21.2\% for $\zeta>1$ behind the wetting front $\mathscr{S}_3$ at $\tau>\tau_s$ (Figure~\ref{fig:discussion}\emph{c}). The perched water table ${\mathscr{S}^*_1}$ rises upwards with constant dimensionless velocity  $\Lambda_{\mathscr{S}^*_1}=-0.268$ until ponding occurs. Meanwhile, the wetting front keeps percolating in the cold firn with reduced velocity $\Lambda_{\mathscr{S}_3}=0.105$. Lastly, the ponding starts at a time when the rising perched water table reaches at the surface so the dimensionless ponding time can be calculated theoretically as $\tau_p = \zeta_t/\Lambda_\mathscr{S} + (-\zeta_t)/\Lambda_{\mathscr{S}^*_1} = 1/0.28 + (-1)/(-0.268)=7.30$. All of these dimensionless shock speeds and times are computed analytically and the resulting locations are graphed with dashed lines in Figures~\ref{fig:discussion}\emph{c}-\emph{e}. The theoretical results show an excellent comparison with the numerical solutions shown in contour plots. The numerical solutions are obtained by solving the governing model \eqref{eq:dimless-system-gov-eqs-1D} performed in the absence of capillary effects in between water and gas phases as well as heat diffusion along with the treatment in the saturated region given by \cite{shadab2022hyperbolic}. Further, the densities of the water and the ice phases are assumed to be the same. The computational domain $\zeta \in [0,2]$ is divided uniformly into 400 cells. The boundary condition at the top surface ($\zeta=0$) is prescribed as the ``Top condition" in Figure~\ref{fig:discussion}\emph{a} whereas the bottom boundary condition is not required. More details about the numerical model are given in \cite{shadab2023mechanism}. This problem constitutes a very specific benchmark test for firn hydrology simulators that are able to simulate variably saturated flows. It shows how vertical heterogeneity in firn may lead to formation of perched aquifers that can cause ponding at a later stage. It also illustrates that an impermeable layer is not required to cause meltwater perching and subsequently, ponding.

\begin{figure}
    \centering
    \includegraphics[width=0.9\linewidth]{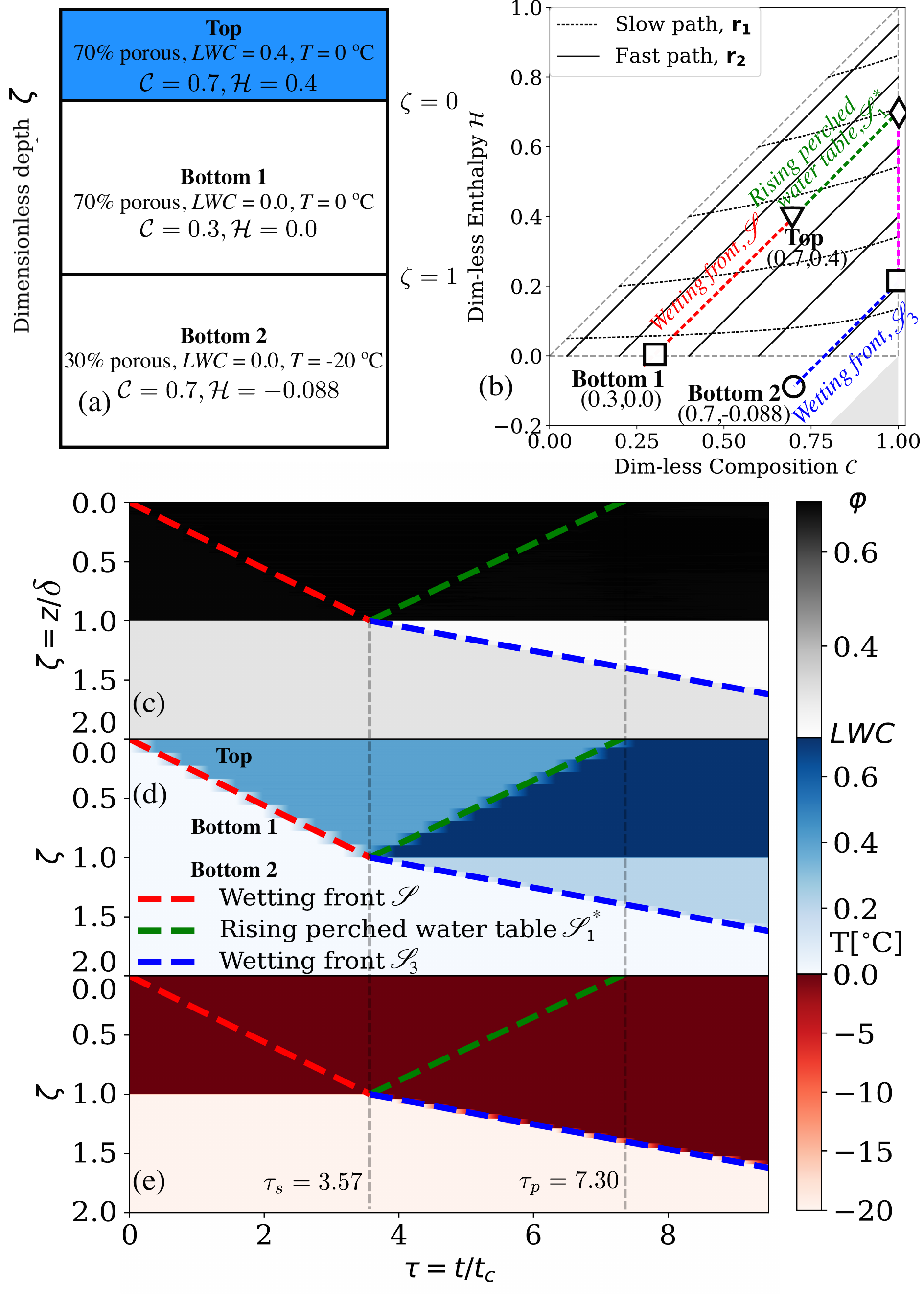}
    \caption{Infiltration into a multilayered firn with porosity and temperature decay with depth: (a) Schematic diagram showing all of the layers (b) Construction of solution in hodograph plane. The contours showing evolution of the firn (c) porosity $\varphi$, (d) liquid water content $LWC$ or volume fraction of water $\phi_w$ and temperature $T$ evaluated by the numerical simulator. Here all dashed lines show analytic solutions computed from the proposed theory. The thin, grey dashed lines show theoretically calculated dimensionless times of saturation $\tau_s$ and ponding $\tau_p$. The theoretical evolution of the initial wetting front $\mathscr{S}$ (red dashed line) is computed from Case III whereas the dynamics of saturated region after wetting front $\mathscr{S}$ reaches $\zeta=1$ shown by blue and green dashed lines is computed by Case XI. Here $\delta$ and $t_c = \delta/K_h$ are characteristic times with former being calculated from their definition \ref{eq:20}. For example, the characteristic depth is $\delta = 5$ m and for $K_h=5.45\times10^{-4}$ m/s (see Table \ref{table:1}), the characteristic time comes out to be $t_c=2.53$ hours. The dimensionless wetting front speeds can be redimensionalized by multiplying with saturated (and no-matrix) hydraulic conductivity $K_h$.}
    \label{fig:discussion}
\end{figure}

\section{Conclusions}\label{sec:conclusion}
This work introduces a unified kinematic wave theory for meltwater infiltration into firn that helps construct analytic solutions in a hodograph plane. The theory neglects heat diffusion and capillary forces while assuming a constant density for water and ice phases. We provide a suite of 12 cases of melt infiltration into firn (except Case VII), inspired by nature, and construct their analytic solution while connecting most of the cases given in the literature. The previous works were predominantly limited to unsaturated wetting fronts in cold and temperate firn. The proposed theory considers the cases that have not been studied previously such as the formation of a perched water table (Case VI and Case XI) and the region of formation of an impermeable ice lens (Case XII). These simple cases can be combined to study a more realistic problem, as was demonstrated by the problem of infiltration into multilayered firn. There are several consequences of this work. First, one can interpret the physics of the meltwater infiltration into firn without running expensive numerical simulations. It will better constrain the process of firn densification and subsequent parameterization for transforming the altimetry data into gravimetry data. Second, these analytic solutions can thus help develop better, cost-effective physics-based firn hydrological models which can then be integrated with ice-sheet and Earth system models. Further, these problems can serve as a benchmark for next generation of wet firn hydrological models which currently show significant deviation due to a lack of benchmark problems. Lastly, from advanced remote sensing and field observation techniques, the proposed cases will help deduce the firn properties such as heterogeneity and percolation parameters such as (relative) permeability exponent. This comprehensive framework can significantly enhance our understanding of wet firn hydrology, a component that has been poorly understood, ultimately aiding in constraining its contribution to surface mass balance loss from glaciers and, consequently, sea-level rise.

\section{Acknowledgements and Funding} 
Support for this work was provided through NASA under Emerging World Grant number NASA $18-$EW$18\_ 2-0027$ and the University of Texas Institute for Geophysics under the Blue Sky Student Fellowship. The authors would like to acknowledge Dr. Surendra Adhikari and Dr. Andreas Colliander for the discussions on meltwater infiltration at the Dye-2 site in Southwest Greenland.

\bibliographystyle{igs}  
\bibliography{igsrefs,marc}   

\appendix

\section{Flux Gradient and Lemma related to eigendecomposition}
The partial derivatives of the flux function given in equation (\ref{eq:dimless-flux-CH-vector}) give:
\begin{align}\label{eq:fCH}
   f_{,\mathcal{H}} &= \frac{\partial f_{ \mathcal{C}}}{\partial \mathcal{H}} = \begin{cases} 0,& \mathcal{H} \leq 0 \\ \mathcal{H}^{n-1} (1-\mathcal{C}+\mathcal{H})^{m-n-1} (n (1-\mathcal{C})+m\mathcal{H}) , & 0 < \mathcal{H}< \mathcal{C}\end{cases} \\
     f_{,\mathcal{C}} &= \frac{\partial f_{\mathcal{H}}}{\partial \mathcal{C}} = \begin{cases} 0,& \mathcal{H} \leq 0 \\ -(m-n)\mathcal{H}^n (1-\mathcal{C}+\mathcal{H})^{m-n-1} , & 0 < \mathcal{H}< \mathcal{C}\end{cases}
\end{align}

\begin{lemma} \textit{Prove that $\lambda_2$ is non-negative and increases monotonically in the direction of integral curves corresponding to the second eigenvector $\textbf{r}_2$ (fast paths) in three phase region, $0<\mathcal{H} < \mathcal{C}$}. \end{lemma}
    \begin{proof} Since $\mathcal{H}>0$, $-\mathcal{C}< \mathcal{H}-\mathcal{C}$, and $0 \leq \mathcal{C} < 1$,
    \begin{align}
        \lambda_2 &=n\mathcal{H}^{n-1} (1-\mathcal{C}+\mathcal{H})^{m-n}
        > 0
    \end{align}
    Along the fast path, $\mathcal{C}=\mathcal{H}+\mathfrak{C}$, 
     \begin{align}
        \lambda_2 &=n\mathcal{H}^{n-1} (1-\mathcal{C}+\mathcal{H})^{m-n}\\
        &= n\mathcal{H}^{n-1} (1-(\mathcal{H}+\mathfrak{C})+\mathcal{H})^{m-n} \\
        &= n\mathcal{H}^{n-1} (1+\mathfrak{C})^{m-n}
    \end{align}   
Now, taking the derivative with respect to $\mathcal{H}$, we get
\begin{align}
    \frac{\d{\lambda}_2}{\d\mathcal{H}} &= n(n-1)\mathcal{H}^{n-2} (1+\mathfrak{C})^{m-n}
\end{align}
As $\mathcal{H}>0$ and $\mathfrak{C}=\mathcal{C}-\mathcal{H}<1$ in three phase region,
\begin{align}
    \frac{\d{\lambda}_2}{\d\mathcal{H}} >0, \quad \textrm{ for } n>1
\end{align}
    Therefore, $\lambda_2$ is non-negative and increases monotonically in the direction of $\textbf{r}_2$ when $n>1$. \end{proof}

    \begin{lemma} \textit{Prove that fast paths are parallel to constant porosity $\varphi$ contours in three phase region, $0<\mathcal{H} \leq \mathcal{C}$.}\end{lemma}
    \begin{proof} In the three-phase region, the porosity $\varphi=\phi_w+\phi_g=\mathcal{H}+(1-\mathcal{C})$. For a constant porosity $\d \varphi=0$, therefore
    \begin{align}
        \d\varphi&=\d(\mathcal{H}+(1-\mathcal{C}))\\
        0&=\d\mathcal{H}-\d \mathcal{C}\\
        1&=\frac{\d\mathcal{H}}{\d\mathcal{C}}\\
\mathcal{C} &= \mathcal{H} + \mathfrak{C} 
    \end{align}
      where $\mathfrak{C}$ again is an integration constant. Therefore, the constant porosity $\varphi$ contours are the integral curves corresponding to second eigenvector $\textbf{r}_2$ (fast paths) in three-phase region. \end{proof}

\end{document}